\documentclass[11pt]{article}

\usepackage{PRIMEarxiv}
\usepackage[colorlinks,linkcolor=magenta,citecolor=blue, pagebackref=true,backref=page]{hyperref}
\usepackage[utf8]{inputenc} 
\usepackage{booktabs}       
\usepackage{amsfonts}       
\usepackage{nicefrac}       
\usepackage{microtype}      
\usepackage{lipsum}
\usepackage{graphicx}       
\usepackage{arxivpackages}
\usepackage{mymacros}
\usepackage{commands}
\usepackage{natbib}
\renewcommand*{\backrefalt}[4]{%
    \ifcase #1 \footnotesize{(Not cited.)}%
    \or        \footnotesize{(Cited on page~#2.)}%
    \else      \footnotesize{(Cited on pages~#2.)}%
    \fi}

\renewcommand{\cite}[1]{\citep{#1}}
\setcitestyle{authoryear,round,citesep={;},aysep={,},yysep={;}}

\graphicspath{{media/}}     

\newcount\Comments  
\title{Convex-Concave Min-Max Stackelberg Games
}

\author{%
  Denizalp Goktas\\
  Department of Computer Science\\
  Brown University\\
  Providence, RI 02912 \\
  \texttt{denizalp\_goktas@brown.edu} \\
   \And
  Amy Greenwald\\
  Department of Computer Science\\
  Brown University\\
  Providence, RI 02912 \\
   \texttt{amy\_greenwald@brown.edu} \\
}

\begin{document}
\maketitle
\begin{abstract}
    Min-max optimization problems (i.e., min-max games) have been attracting a great deal of attention because of their applicability to a wide range of machine learning problems.
    Although significant progress has been made recently, the literature to date has focused on games with independent action sets; little is known about solving games with dependent action sets, which can be interpreted as min-max Stackelberg games, i.e., sequential two-player zero-sum games. The canonical solution concept for min-max Stackelberg games is the Stackelberg equilibrium, whose existence we establish when the objective function is continuous and the constraints satisfy appropriate convexity conditions. 
    We then introduce two first-order methods that compute Stackelberg equilibria in a large class of convex-concave min-max Stackelberg games, and show that our methods converge in polynomial time. 
    Min-max Stackelberg games were first studied by Wald, under the posthumous name of Wald's maximin model, a variant of which is the main paradigm used in robust optimization, which means that our methods can likewise be used to solve many robust convex optimization problems.
    We observe that the computation of competitive equilibria in homothetic Fisher markets also comprises a min-max Stackelberg game. Further, we demonstrate the efficacy and efficiency of our algorithms in practice by computing competitive equilibria in homothetic Fisher markets with varying utility structures. Our experiments suggest potential ways to extend our theoretical results, by demonstrating how different smoothness properties can affect the convergence rate of our algorithms.
\end{abstract}

\section{Introduction}
\label{sec:intro}

Min-max optimization problems have attracted a great deal of attention recently because of their applicability to a wide range of machine learning problems.
Examples of settings in which min-max optimization problems arise include, but are not limited to, reinforcement learning \cite{dai2018rl}, generative adversarial networks 
\cite{goodfellow2020generative,sanjabi2018convergence}, fairness in machine learning \cite{dai2019kernel, edwards2016censoring, madras2018learning, sattigeri2018fairness, xu2018fairgan}, adversarial learning \cite{sinha2020certifying}, generative adversarial imitation learning \cite{cai2019global, hamedani2018iteration}, and statistical learning (e.g., learning parameters of exponential families) \cite{dai2019kernel}.
These applications often require solving a \mydef{constrained min-max optimization problem} 
(with independent constraint sets), i.e., $\min_{\outer \in \outerset} \max_{\inner \in \innerset} \obj(\outer, \inner)$.

A \mydef{convex-concave} constrained min-max optimization problem is one in which $\obj$ is convex in $\outer$ and concave in $\inner$.
In the special case of convex-concave objective functions, the seminal minimax theorem holds: i.e.,
$\min_{\outer \in \outerset} \max_{\inner \in \innerset} \obj(\outer, \inner) = \max_{\inner \in \innerset} \min_{\outer \in \outerset} \obj(\outer, \inner)$ \cite{neumann1928theorie}.
This theorem guarantees the existence of a \mydef{saddle point}, i.e., a point $(\outer^*, \inner^*) \in \outerset \times \innerset$ s.t.\ for all $\outer \in \outerset$ and $\inner \in \innerset$, $\obj(\outer^*, \inner) \leq \obj(\outer^*, \inner^*) \leq \obj(\outer, \inner^*)$.

As a saddle point is simultaneously a minimum of $\obj$ in the $\outer$-direction and a maximum of $\obj$ in the $\inner$-direction, we can interpret such optimization problems as simultaneous-move zero-sum games between an $\outer$- and $\inner$-player with respective action sets $\outerset$, $\innerset$, and payoff functions $- \obj$, $\obj$, in which case $\inner^*$ (resp.\ $\outer^*$) can be interpreted as a best-response of the $\outer$- (resp.\ $\inner$-) player to their opponent's action $\outer^*$ (resp.\ $\inner^*)$.
As such, the optimization problem is called \mydef{a convex-concave min-max (simultaneous-move) game},
and any saddle point is also called a minimax point or a \mydef{Nash equilibrium}.

In this paper, we show that a competitive equilibrium in a Fisher market, the canonical solution of a well-studied market model in algorithmic game theory \cite{AGT-book}, can be understood as a solution to a convex-concave min-max optimization problem, albeit with \emph{dependent\/} constraint sets.
Formally, we define a \mydef{constrained min-max optimization problem with dependent constraint sets} as an optimization problem of the following form: $\min_{\outer \in \outerset} \max_{\inner \in \coupledset(\outer)} \obj(\outer, \inner)$, where the \mydef{objective} $\obj: \outerset \times \innerset \to \R$ is continuous, the \mydef{constraint set} $\outerset \subset \R^\outerdim$ is non-empty and compact, and the \mydef{constraint correspondence} $\coupledset: \outerset \rightrightarrows \innerset \subset \R^\innerdim$ is non-empty, compact-valued, and continuous (i.e., upper and lower hemicontinuous). 
As is standard in the literature (see, for instance, \citet{facchinei2009generalized}), we represent the constraint correspondence $\coupledset(\outer) \doteq \left\{ \inner \in \innerset \mid \constr (\outer, \inner) \ge \zeros \right\}$ by \mydef{coupling constraints} $\constr (\outer, \inner) \doteq \left(\constr[1](\outer, \inner), \hdots, \constr[\numconstrs](\outer, \inner) \right)^T$ with $\constr[\numconstr]: \R^\outerdim \times \R^\innerdim \to \R$, for all $\numconstr \in [\numconstrs]$, and often write the optimization problem as $\min_{\outer \in \outerset} \max_{\inner \in \innerset: \constr(\outer, \inner) \geq \zeros} \obj(\outer, \inner)$.

Although it is known that solutions to convex-concave min-max optimization problems with independent constraints exist when the objective $\obj$ is continuous and the constraint sets are non-empty and compact, existence of a solution in the dependent setting further requires the constraint correspondence to be continuous, which does not directly follow from the coupling constraints being continuous, but rather requires additional constraint qualifications set on the coupling constraints (see \Cref{stackleberg-existence} and \Cref{sec:assum_and_exist}).

Perhaps more importantly, while in the independent setting the minimax theorem \cite{neumann1928theorie} is guaranteed to hold when the objective is convex-concave and the constraint sets are convex, i.e., $\min_{\outer \in \outerset} \max_{\inner \in \innerset} \obj(\outer, \inner) =  \max_{\inner \in \innerset} \min_{\outer \in \outerset} \obj(\outer, \inner)$,
a minimax theorem does not hold in the dependent setting, i.e., $\min_{\outer \in \outerset} \max_{\inner \in \innerset: \constr(\outer, \inner) \geq \zeros} \obj(\outer, \inner) \neq  \max_{\inner \in \innerset} \min_{\outer \in \outerset: \constr(\outer, \inner) \geq \zeros} \obj(\outer, \inner)$, even when $\constr[\numconstr]$ is affine, for all $\numconstr \in [\numconstrs]$:

\begin{example}
\label{min-max-fail-ex}
Consider the following constrained min-max optimization problem with dependent constraint sets:
$\min_{\outer[ ] \in [-1, 1] } \max_{\inner[ ] \in [-1, 1] : \outer[ ] + \inner[ ] \leq 0} \outer[ ]^2 + \inner[ ] + 1$.
The optimum is $\outer[ ]^{*} = \nicefrac{1}{2}, \inner[ ]^{*} = - \nicefrac{1}{2}$, with value $\nicefrac{3}{4}$.
Now, consider the same problem, with the order of the $\min$ and the $\max$ reversed:
$\max_{\inner[ ] \in [-1, 1]} \min_{\outer[ ] \in [-1, 1] : \outer[ ] + \inner[ ] \leq 0} \outer[ ]^2 + \inner[ ] + 1$.
The optimum is now $\outer[ ]^{*} = -1, \inner[ ]^{*} = 1$, with value $3$.
\end{example}

Without a minimax theorem, a constrained min-max optimization problem in the dependent setting cannot be interpreted as as a \emph{simultaneous-move\/} (pseudo-)game, nor its solutions as (generalized) Nash equilibria:%
\footnote{Technically speaking, settings in which the players' actions are collectively constrained are not games but pseudo-games; see~\Cref{sec-app:GNE}.}
Instead, they are more appropriately viewed as sequential zero-sum, i.e., \mydef{min-max Stackelberg}, 
games where the $\outer$-player (or the leader) chooses $\outer \in \outerset$ before the $\inner$-player (or the follower) responds with their choice of $\hat{\inner} (\outer) \in \innerset$ s.t.\ $\constr(\outer, \hat{\inner} (\outer)) \geq \zeros$.
The relevant equilibrium concept is then a Stackelberg equilibrium \cite{stackelberg1934marktform},
in which the $\outer$-player optimizes their choice assuming the $\inner$-player will best-respond: i.e., optimize their choice in turn.
We thus refer to constrained min-max optimization problems with dependent constraint sets as \mydef{min-max Stackelberg games}.

For such games, we define the \mydef{value function}%
\footnote{Note that this use of the term value function comes from economics, and is distinct from its use in reinforcement learning.}
$\val: \outerset \to \R$ as $\val(\outer) = \max_{\inner \in \innerset : \constr(\outer, \inner) \geq \zeros} \obj(\outer, \inner)$.
This function represents the $\outer$-player's loss, assuming the $\inner$-player chooses a feasible best-response, so it is the function the $\outer$-player seeks to minimize.
The $\inner$-player seeks only to maximize the objective function, given the action of the $\outer$-player, subject to the coupling constraints.



In min-max simultaneous-move games, the assumption that the objective function is convex-concave implies that the function that the $\outer$-player (resp. $\inner$-player) seeks to optimize is convex (concave) in its action. 
Correspondingly, a \mydef{convex-concave min-max Stackelberg game} is more appropriately defined as one where $\val (\outer)$ is convex in $\outer$ and $\obj (\outer, \inner)$ is concave in $\inner$, for all $\outer \in \outerset$.%
\footnote{We present sufficient conditions for a min-max Stackelberg game to be convex-concave in \Cref{sec:assum_and_exist}.}
The first-order necessary and sufficient conditions for a tuple $(\outer^*, \inner^*) \in \outerset \times \innerset$ to be a Stackelberg equilibrium in a convex-concave min-max Stackelberg game are given by the KKT stationarity conditions for the two players' optimization problems, namely $\min_{\outer \in \outerset} \val (\outer)$ for the $\outer$-player, and $\max_{\inner \in \innerset : \constr (\outer^*, \inner) \geq \zeros} \obj (\outer^*, \inner)$ for the $\inner$-player.

In the independent constraint set---hereafter, action set---setting, Danskin's theorem~\cite{danskin1966thm} states that $\grad[\outer] \val(\outer) = \grad[\outer] \obj(\outer, \inner^*(\outer))$, where $\inner^*(\outer) \in \argmax_{\inner \in \innerset} \obj(\outer, \inner)$.
In other words, when there is no dependence among the players' action sets, the gradient of the value function coincides with that of the objective function.
The first-order necessary and sufficient conditions for a tuple $(\outer^*, \inner^*)$ to be an (interior) saddle point are for it to be a stationary point of $\obj$ (i.e., $\grad[\outer] \obj(\outer^*, \inner^*) = \grad[\inner] \obj(\outer^*, \inner^*) = 0$).
It is therefore logical for players in the independent/simultaneous-move setting to follow the gradient of the objective function.
In the dependent/sequential setting, however, the direction of steepest descent (resp.\ ascent) for the outer (resp.\ inner) player is the gradient of their \emph{value\/} function.

\begin{example}
\label{grad-fail-ex}
Consider $\min_{x \in [-1,1]} \max_{y \in [-1,1] : x + y \leq 0} x^2 + y + 1$,
and recall \citeauthor{jin2020local}'s gradient descent with max-oracle algorithm \cite{jin2020local}:
$\outer^{(t+1)} = \outer^{(t)} - \learnrate[ ] \grad[\outer] \obj \left( \outer^{t}, \inner^{*} (\outer^{(t)}) \right)$,
where $\inner^{*} (\outer^{(t)}) \in \argmax_{\inner \in \innerset : \constr(\outer^{(t)}, \inner) \geq 0} \obj(\outer^{(t)}, \inner)$ for $\eta > 0$.
Applied to this sample problem, with $\eta = 1$, this algorithm yields the following update rule:
$\outer[ ]^{(t+1)} = \outer[ ]^{(t)} - 2\outer[ ]^{(t)} = - \outer[ ]^{(t)}$.
Thus, letting $\outer[ ]^{(0)}$ equal any feasible $\outer[ ]$,
the output cycles between $\outer[ ]$ and $-\outer[ ]$, so that the average of the iterates converges to $\outer[ ]^{*} = 0$ (with $\inner[ ]^{*} = 0$), which is not a Stackelberg equilibrium, as the Stackelberg equilibrium of this game is $\outer[ ]^{*} = \nicefrac{1}{2}, \inner[ ]^{*} = - \nicefrac{1}{2}$.


Now consider an algorithm that updates based not on gradient of the objective function, but 
of the value function, namely
$\outer^{(t+1)} = \outer^{(t)} - \learnrate[ ][ ] \grad[\outer] \val \left( \outer^{t} \right)$.
The value function is $\val(\outer[ ]) = \max_{\inner[ ] \in [-1, 1] : \outer[ ] + \inner[ ] \leq 0} \outer[ ]^2 + \inner[ ] + 1 = \outer[ ]^2 - \outer[ ] + 1$, with gradient $\val^{\prime}(\outer[ ]) =  2 \outer[ ] - 1$.
Thus, when $\eta = 1$, this algorithm yields the following update rule:
$\outer[ ]^{(t+1)} = \outer[ ]^{(t)} - \val^{\prime} (\outer[ ]^{(t)}) = \outer[ ]^{(t)} - 2 \outer[ ]^{(t)} + 1 = - \outer[ ]^{(t)} + 1$.
If we run this algorithm from initial point $\outer[ ]^{(0)} = \nicefrac{1}{8}$,
    we get 
    $\outer[ ]^{(1)} = - \nicefrac{1}{8} + 1 = \nicefrac{7}{8}$, $\outer[ ]^{(2)} = - \nicefrac{7}{8} + 1 = \nicefrac{1}{8}$, and so on.
    The average of the iterates $\nicefrac{1}{8}, \nicefrac{7}{8}, \nicefrac{1}{8}, \hdots$ converges to $\outer[ ]^{*} = \nicefrac{1}{2}\left( \nicefrac{1}{8} + \nicefrac{7}{8}\right) =  \nicefrac{1}{2}$, and correspondingly $\inner[ ]^{*} = - \nicefrac{1}{2}$, which is indeed the Stackelberg equilibrium.
\end{example}

\subsection{Contributions}

In this paper, we study first-order methods to compute Stackelberg equilibria in min-max Stackelberg games.
To this end, we start by presenting sufficient conditions (\Cref{assum:existence}) for the existence of Stackelberg equilibrium (\Cref{stackleberg-existence}), namely: the objective $\obj$ and the constraint $\constr$ functions are continuous, the actions sets $\outerset, \innerset$ are non-empty and compact, with $\innerset$ additionally convex, and the coupling constraints $\constr$ are quasi-concave and satisfiable.
This result makes weaker assumptions than known existence results for Stackelberg equilibrium (see, for example, \citet{Lucchetti1981SEexist}), which require uniqueness of the follower's best response.%
\footnote{Note that our conditions do not guarantee existence in general-sum Stackelberg games, where non-uniqueness of the follower's best response leads to competing definitions of Stackelberg equilibria, namely strong (resp.\@ weak), when ties are broken in favor (resp.\@ to the detriment) of the leader \cite{conitzer2006computing}. \amy{we should eventually add a reference here to the recent neurips paper, if accepted, which relaxes the uniqueness assumption even in the general-sum case.}}


We then introduce two first-order (subgradient) methods that solve min-max Stackelberg games---to our knowledge the first such methods.
Our approach relies on a new generalization of a series of fundamental results in mathematical economics known as \mydef{envelope theorems} \cite{afriat1971envelope,milgrom2002envelope}.
Envelope theorems generalize aspects of Danskin's theorem,
by providing explicit formulas for the gradient of the value function in dependent action settings,
\emph{when a derivative is guaranteed to exist}.
To sidestep the differentiability issue, we introduce a generalized envelope theorem that gives an explicit formula for the subdifferential of the value function at any point $\outer \in \outerset$ in dependent action settings with a convex value function $\val$.

Our first algorithm follows \citet{jin2020local},
assuming access to a max-oracle that returns $\inner^*(\outer) \in 
\argmax_{\inner \in \innerset: \constr(\outer, \inner) \geq 0} \obj(\outer, \inner)$, given $\outer \in \outerset$.
Hence, our first algorithm solves only for an optimal $\outer^{*}$, while our second algorithm explicitly solves for both $\outer^{*}$ and $\inner^{*}(\outer^{*})$.
We show that under suitable assumptions (\Cref{main-assum}) both algorithms converge in polynomial time to a Stackelberg equilibrium of any convex-concave min-max Stackelberg game (\Cref{l-mogd} and \Cref{ngd-thms}, respectively).
In \Cref{tab:our_results}, we summarize the iteration complexities of our algorithms, i.e., the number of iterations required to achieve an $\varepsilon$-approximate equilibrium, where $\varepsilon$ is the desired precision. 



\vspace{-3.33mm}
\begin{table}[H]
    \caption{Iteration complexities of Algorithms \ref{mogd} and \ref{ngd} for min-max Stackelberg games. We call a min-max Stackelberg game $\mu_\outer$-strongly-convex-$\mu_\inner$-strongly-concave if $\val(\outer)$ is $\mu_\outer$-strongly-convex in $\outer$, and $\obj(\outer, \inner)$ if $\mu_\inner$-strongly-concave for all $\outer \in \outerset$. Here, $\mu_\outer$ and $\mu_\inner$ are strong convexity/concavity parameters.}
    \label{tab:our_results}
    \centering
    \renewcommand*\arraystretch{1.33}
    \begin{tabular}{|c|c|c|}\hline
        \multirow{2}{*}{Properties of the min-max Stackelberg game} & \multicolumn{2}{c|}{Iteration Complexity} \\ \cline{2-3}
        & Algorithm \ref{mogd} & Algorithm \ref{ngd}  \\ \hline
        $\mu_\outer$-strongly-convex-$\mu_\inner$-strongly-concave &  \multirow{2}{*}{$O\left(\varepsilon^{-1}\right)$} & $\tilde{O}(\varepsilon^{-1})$ \\\cline{1-1} \cline{3-3} 
        $\mu_\outer$-strongly-convex-concave &  & $O(\varepsilon^{-2})$ \\ 
        \hline \hline 
        convex-$\mu_\inner$-strongly-concave & \multirow{2}{*}{$O\left(\varepsilon^{-2}\right)$} & $\tilde{O}(\varepsilon^{-2})$\\ \cline{1-1} \cline{3-3}
        convex-concave &   & $O(\varepsilon^{-3})$\\ \hline 
    \end{tabular}
    \renewcommand*\arraystretch{1}
\end{table}
Finally, we apply our results to the computation of competitive equilibria in Fisher markets.
In this context, our method for solving a min-max Stackelberg game reduces to solving the market in a decentralized manner using a natural market dynamic called \emph{t\^atonnement} \cite{walras}.
We demonstrate the efficacy and efficiency of our algorithms 
in practice by running a series of experiments in which we compute competitive equilibria in Fisher markets with varying utility structures---specifically, linear, Cobb-Douglas, and Leontief.
Although our theoretical results do not apply to all these Fisher markets---Leontief utilities, in particular, are not differentiable---\emph{t\^atonnement} converges in all our experiments.
That said, the rate of convergence does seem to depend on the smoothness characteristics of the utility structures; we observe slower convergence for Leontief utilities, and faster convergence than our theory predicts for Cobb-Douglas utilities, which are not only differentiable, but whose value function is also differentiable.

\subsection{Related Work}


Our model of min-max Stackelberg games seems to have first been studied by \citeauthor{wald1945maximin}, under the posthumous name of Wald's maximin model~\cite{wald1945maximin}.
A variant of Wald's maximin model is the main paradigm used in robust optimization, a fundamental framework in operations research for which many methods have been proposed \cite{ben2015oracle, ho2018online, postek2021firstorder}. 
\citet{shizimu1981stackelberg,shimizu1980minmax} proposed the first algorithm to solve min-max Stackelberg games
via a relaxation to a constrained optimization problem with infinitely many constraints, which nonetheless seems to perform well in practice.
More recently, \citet{segundo2012evolutionary} proposed an evolutionary algorithm for these games, but they provided no guarantees.
As pointed out by \citeauthor{postek2021firstorder}, all prior methods either require oracles and are stochastic in nature \cite{ben2015oracle}, or rely on a binary search for the optimal value, which can be computationally complex \cite{ho2018online}.
The algorithms we propose in this paper circumvent the aforementioned issues and can be used to solve a large class of robust convex optimization problems in a simple and efficient manner.


Extensive-form games in which players' action sets can depend on other players' actions have been studied by \citet{davis2019solving} assuming payoffs are bilinear, and by \citet{farina2019regret} for another specific class of convex-concave payoffs.
\citet{fabiani2021local} and  \citet{kebriaei2017discrete} study more general settings than ours, namely general-sum Stackelberg games with more than two players.
Both sets of authors derive convergence guarantees
assuming specific payoff structures, but their algorithms do not converge in polynomial time.


Min-max Stackelberg games naturally model various economic settings. 
They are related to abstract economies, first studied by \citet{arrow-debreu}; however, the solution concept par excellence for abstract economies is generalized Nash equilibrium \cite{facchinei2007generalized, facchinei2010generalized}, which, like Stackelberg, is a weaker solution concept than Nash, but which makes the arguably unreasonable assumption that the players move simultaneously and nonetheless satisfy the constraint dependencies on their actions imposed by one another's moves.
See \Cref{sec-app:GNE} for a more detailed discussion of generalized Nash equilibria versus Stackelberg equilibria. 

\citet{duetting2019optimal} study optimal auction design problems.
They propose a neural network architecture called RegretNet that represents optimal auctions, and train their networks using a variant of Algorithm~\ref{ngd-thms}.
Optimal auction design problems can be seen as min-max Stackelberg games;
however, as their objectives are non-convex-concave in general, our guarantees do not apply.

In this paper, we observe that solving for the competitive equilibrium of a Fisher market can also be seen as solving a convex-concave min-max Stackelberg game.
The study of the computation of competitive equilibria in Fisher markets was initiated by \citet{devanur2002market}, who provided a polynomial-time method for the case of linear utilities.
\citet{jain2005market} subsequently showed that a large class of Fisher markets could be solved in polynomial-time using interior point methods.
Recently, \citet{gao2020polygm} studied an alternative family of first-order methods for solving Fisher markets (only; not min-max Stackelberg games more generally), assuming linear, quasilinear, and Leontief utilities, as such methods can be more efficient when markets are large.

See \Cref{sec-app:related} for a detailed discussion of recent progress on solving min-max Stackelberg games,
both in the convex-concave case and the non-convex-concave case.

\section{Preliminaries}
\label{sec:prelim}

\paragraph{Notation} 
We use Roman uppercase letters to denote sets (e.g., $X$),
bold uppercase letters to denote matrices (e.g., $\allocation$), bold lowercase letters to denote vectors (e.g., $\price$), and Roman lowercase letters to denote scalar quantities, (e.g., $c$).
We denote the $i$th row vector of a matrix (e.g., $\allocation$) by the corresponding bold lowercase letter with subscript $i$ (e.g., $\allocation[\buyer])$.
Similarly, we denote the $j$th entry of a vector (e.g., $\price$ or $\allocation[\buyer]$) by the corresponding Roman lowercase letter with subscript $j$ (e.g., $\price[\good]$ or $\allocation[\buyer][\good]$).
We denote the set of integers $\left\{1, \hdots, n\right\}$ by $[n]$, the set of natural numbers by $\N$, the set of real numbers by $\R$, the set of non-negative real numbers by $\R_+$. We denote by $\simplex[n] = \{\x \in \R_+^n \mid \sum_{i = 1}^n x_i = 1\}$. We denote norms by $\| \cdot \|$, and unless otherwise noted we assume that all norms are Euclidean, i.e., $\| \cdot \| = \| \cdot \|_2$.
We denote by $\project[\innerset]$ the Euclidean projection operator onto the set $\innerset \subset \R^{\innerdim}$: i.e., $\project[\innerset](\inner) = \argmin_{\z \in \innerset} \left\| \inner - \z \right\|_2$.
Given two vectors $\x, \y \in \R^n$, we write $\x \geq \y$ or $\x > \y$ to mean component-wise $\ge$ or $>$, respectively. 
The linear composition of two sets $\outerset, \innerset \subset \R^d$ for any constant $a, b \in \R$, is given by their Minkowski sum and product, i.e., $a\outerset + b\innerset \doteq \{a\outer + b\inner \mid \outer \in \outerset, \inner \in \innerset \}$.

\paragraph{Mathematical Definitions}

We now define several mathematical concepts that are used in our convergence proofs.
Let $\left\| \cdot \right\|: \calX \to \R_+$ be any norm.%
\footnote{Throughout this paper, unless otherwise noted, we assume $\left\| \cdot \right\|$ is the Euclidean norm, i.e., $\left\| \cdot \right\| = \left\| \cdot \right\|_2$.}
Given $A \subset \R^\outerdim$, the function $\obj: A \to \R$ is said to be $\lipschitz[\obj]$-\mydef{Lipschitz-continuous} w.r.t.\ norm $\left\| \cdot \right\|$ iff $\forall \outer_1, \outer_2 \in X, \left\| \obj(\outer_1) - \obj(\outer_2) \right\| \leq \lipschitz[\obj] \left\| \outer_1 - \outer_2 \right\|$.
If the gradient of $\obj$, $\grad \obj$, is $\lipschitz[\grad \obj]$-Lipschitz-continuous, we refer to $\obj$ as $\lipschitz[\grad \obj]$-\mydef{Lipschitz-smooth}.
A function $\obj: A \to \R$ is $\mu$-\mydef{strongly convex} if $\obj(\outer_1) \geq \obj(\outer_2) + \left< \grad[\outer] \obj(\outer_2), \outer_1 - \outer_2 \right> + \nicefrac{\mu}{2} \left\| \outer_1 - \outer_1 \right\|^2$, and $\mu$-\mydef{strongly concave} if $-\obj$ is $\mu$-strongly convex.

\subsection{Min-Max Stackelberg Games}

\mydef{A min-max Stackelberg game}, denoted $(\outerset, \innerset, \obj, \constr)$, is a two-player, zero-sum game, where one player, who we call the $\outer$-player 
first commits to an action $\outer \in \outerset$ from its action space $\outerset \subset \R^\outerdim$, after which a second player called the $\inner$-player, takes an action from its action set $\coupledset(\outer) \subset \innerset$ given by an action correspondence $\coupledset(\outer) \doteq  \left\{ \inner \in \innerset \mid \constr (\outer, \inner) \geq \zeros \right\}$ defined by \mydef{coupling constraints} $\constr (\outer, \inner) = \left(\constr[1](\outer, \inner), \hdots, \constr[\numconstrs](\outer, \inner) \right)^T$ with $\constr[\numconstr]: \R^\outerdim \times \R^\innerdim \to \R$, for all $\numconstr \in [\numconstrs]$.
An action profile $(\outer, \inner) \in \outerset \times \innerset$ is said to be \mydef{feasible} iff $\constr (\outer, \inner) \ge \bm{0}$.
The function $\obj$ maps a pair of feasible strategies taken by the players $(\outer, \inner) \in \outerset \times \innerset$ to a real value (i.e., a payoff), which represents the loss (resp.\ the gain) of the $\outer$-player (resp.\ $\inner$-player).

We define the \mydef{value function} $\val: \outerset \to \R$ as $\val(\outer) = \max_{\inner \in \innerset: \constr(\outer, \inner) \geq \zeros}$.
A min-max Stackelberg game is said to be \mydef{convex-concave} if the objective function $\obj$ is concave in $\inner$ for all $\outer \in \outerset$, and $\val$ is convex in $\outer$. 
Note that while a min-max game with a convex-concave objective function $\obj$ is a convex-concave min-max Stackelberg game, the converse is not true.%
\footnote{For example, the non-convex-concave min-max optimization problem: $\min_{\outer[ ] \in [-1, 1]} \max_{\inner[ ] \in [-1, 1]} \outer[ ]^3 \inner[ ] - \frac{1}{2} \outer[ ]^2 \inner[ ]^2$ has convex value function $\val(\outer[ ]) = \frac{1}{2} \outer[ ]^4$, because $\inner[ ]^*(\outer[ ]) \in \argmax_{\inner[ ] \in [-1, 1]} \outer[ ]^3 \inner[ ] - \frac{1}{2} \outer[ ]^2 \inner[ ]^2 = \{\outer[ ]\}$.
The results in this paper imply that solving for Stackelberg equilibria in a class of non-convex-concave min-max games (with independent constraints) is tractable.}

The relevant solution concept for Stackelberg games is the \mydef{Stackelberg equilibrium}:

\begin{definition}[Stackelberg Equilibrium]
Consider the min-max Stackelberg game $(\outerset, \innerset, \obj, \constr)$. An action profile $\left(\outer^{*}, \inner^{*} \right) \in \outerset \times \innerset$ such that $\constr(\outer^{*}, \inner^{*}) \geq \zeros$ is a $(\varepsilon, \delta)$-Stackelberg equilibrium if
\begin{align*}
\max_{\inner \in \innerset : \constr(\outer^{*}, \inner) \geq 0} \obj \left( \outer^{*}, \inner \right) - \delta \leq \obj \left( \outer^{*}, \inner^{*} \right) \leq \min_{\outer \in \outerset} \max_{\inner \in \innerset : \constr(\outer, \inner) \geq 0} \obj \left( \outer, \inner \right) + \varepsilon.
\end{align*}
\end{definition}

Intuitively, a $(\varepsilon, \delta)$-Stackelberg equilibrium is a point at which the $\outer$-player's (resp.\ $\inner$-player's) payoff is no more than $\varepsilon$ (resp.\ $\delta$) away from its optimum.
A $(0,0)$-Stackelberg equilibrium is simply called a \mydef{Stackelberg equilibrium}.

\section{Assumptions and Existence}\label{sec:assum_and_exist}


While the continuity of $\obj$ together with a non-empty, compact action set $\innerset$ is sufficient to guarantee the existence of a solution to the $\inner$-player's maximization problem, it is not sufficient to guarantee continuity of the value function.
As a result, it is likewise not sufficient to guarantee existence of a Stackelberg equilibrium, even when $\outerset$ is also non-empty and compact. 
If the constraint correspondence $\coupledset$ is continuous, however, then the Maximum Theorem \cite{berge1997topological} guarantees the continuity of the value function $\val$, and in turn the existence of a Stackelberg equilibrium with a unique value by the Weierstrass extreme value theorem \cite{protter2012extreme}.
We specialize this result to the natural representation of the constraint correspondence $\coupledset$ in terms of coupling constraints $\constr$, by presenting assumptions on the coupling constraints $\constr$ under which the value function is again guaranteed to be continuous, which again implies the existence of Stackelberg equilibria, all with a unique value, in min-max Stackelberg games.
Note that Stackelberg equilibria are not guaranteed to exist under similar assumptions in general-sum Stackelberg games \cite{Lucchetti1981SEexist}.
(Nor are they guaranteed to be uniquely valued.)

\begin{assumption}[Existence Assumptions]
\label{assum:existence}
1.~$\outerset, \innerset$ are non-empty and compact, with $\innerset$ additionally convex;
2.~$f, \constr[1], \ldots, \constr[\numconstrs]$ are continuous in $(\outer, \inner)$; and
3.~$\constr[1], \ldots, \constr[\numconstrs]$ are quasi-concave in $(\outer, \inner)$ and satisfiable, i.e., for all 
$\outer \in \outerset$, there exists $\widehat{\inner} \in \innerset$ s.t.\ $\constr(\outer, \widehat{\inner}) \geq \zeros$.

\end{assumption}

\begin{restatable}{theorem}{stacklebergexistence}
\label{stackleberg-existence}
If $(\outerset, \innerset, \obj, \constr)$ is a min-max Stackelberg game that satisfies \Cref{assum:existence}, a Stackelberg equilibrium exists, and its value is unique.
\end{restatable}

The proofs of all our results appear in the appendix.

Part 3 of \Cref{assum:existence} can alternatively be replaced by the following two conditions (see Example 5.10 of \citet{rockafellar2009variational}):
3a.~(Slater's condition) 
$\forall \outer \in \outerset, \exists \widehat{\inner} \in \innerset$ s.t.\ $g_{\numconstr}(\outer, \widehat{\inner}) > 0$, for all $\numconstr \in [\numconstrs]$;
3b.~$ \constr[1], \hdots, \constr[\numconstrs]$ are concave in $\inner$, for all $\outer \in \outerset$ . 
This fact will be useful when analyzing the algorithms we propose to compute Stackelberg equilibria in min-max Stackelberg games (\Cref{sec-app:algos}).
In particular, the convergence results in this paper apply under the following necessary and sufficient assumptions, which guarantee the existence of a Stackelberg equilibrium, and are satisfied by the two examples presented thus far, as well as by homothetic Fisher markets (\Cref{sec:fisher}):

\begin{assumption}[Convergence Assumptions]
\label{main-assum}
1.~$\outerset, \innerset$ are non-empty, compact, and convex; 
2.~$\grad[\outer] f, \grad[\outer] \constr[1], \ldots, \grad[\outer] \constr[\numconstrs]$ are continuous;
3.~(Slater's condition) 
$\forall \outer \in \outerset, \exists \widehat{\inner} \in \innerset$ s.t.\ $g_{\numconstr}(\outer, \widehat{\inner}) > 0$, for all $\numconstr \in [\numconstrs]$;
4.~the value function $\val(\outer)$ is convex in $\outer$; and
5.~$\obj, \constr[1], \hdots, \constr[\numconstrs]$ are concave in $\inner$, for all $\outer \in \outerset$.
\end{assumption}

Part 1 of \Cref{main-assum} ensures that the projection of any actions onto $\outerset$ and $\innerset$ is unique and well-defined; furthermore, the convexity of $\outerset$ is sufficient, though not necessary, for the existence of a Stackelberg equilibrium.
Part 2 can be replaced by a subgradient boundedness assumption: i.e., Lipschitz-continuity; 
nonetheless, for simplicity, we assume this stronger condition.
Part 3, Slater's condition, is standard in the convex programming literature \cite{boyd2004convex};
it is a constraint qualification condition, which ensures that the value function $\val$ is Lipschitz-continuous, thereby allowing us to derive polynomial-time first-order methods%
\footnote{Without Slater's condition, the problem becomes analytically intractable, as the value function cannot be expressed via a Lagrangian relaxation, since optimal KKT multipliers are not guaranteed to exist, making it hard, if not impossible, to obtain an analytical expression for the derivative of the value function, even at points at which it is differentiable.}
for computing Stackelberg equilibria in min-max Stackelberg games, when combined with 
Parts 4 and 5 of \Cref{main-assum} are necessary%
\footnote{Technically speaking, for Parts 4 and 5 of \Cref{main-assum} to be necessary, one would need to replace convexity and concavity with invexity and incavity; however, as these properties are rarely invoked in the game theory and optimization literature, we formulate the assumption as such.}
and sufficient to ensure the
problem is tractable via first-order methods; without it the problem is NP-hard \cite{tsaknakis2021minimax}. 


Under \Cref{main-assum}, the set of Stackelberg equilibria of a min-max Stackelberg game satisfies several desirable mathematical properties:

\begin{restatable}{proposition}{propstackelbergconvex}\label{prop:stackelberg_convex}
Under \Cref{main-assum}, the set of Stackelberg equilibria of any min-max Stackelberg game $(\outerset, \innerset, \obj, \constr)$ is non-empty, compact, and convex.
\end{restatable}

Observe that Part 4 of \Cref{main-assum} is not a first-order assumption.
Obtaining first-order necessary and sufficient conditions for the convexity of $\val$ seems out of reach, even when the action sets are not dependent, e.g., $\constr$ is the zero function, as this would require one to derive necessary and sufficient conditions for the intersection of two sets to be convex, which is an open question \cite{rockafellar2009variational}.
Consequently, we also provide the following \emph{two\/} first-order assumptions as alternatives, each of which is sufficient to ensure the convexity of $\val$.
The examples and applications in this paper satisfy both of these alternative assumptions.

\begin{assumption}[Alternative 1 to Part 4 of Assumption \ref{main-assum}]
\label{assum:alter_1}
4$'$a.~(Convex objective)~$\obj$ is convex in ($\outer, \inner$); and
4$'$b.~(Concave constraint correspondence) $\coupledset$ is concave, i.e., $\outer, \outer^\prime \in \outerset, \lambda \in (0, 1)$, $\coupledset(\lambda \outer + (1-\lambda)\outer^\prime) \subseteq \lambda \coupledset(\outer) +  (1-\lambda)\coupledset(\outer^\prime)$.
\end{assumption}

The proof that \Cref{assum:alter_1} guarantees convexity of the value function $\val$ can be found in Proposition 2.7 of \citet{fiacco1986convexity}. 
Although it is well-known that one can guarantee the \emph{convexity\/} of the constraint correspondence $\coupledset$ (i.e., for all $\outer, \outer^\prime \in \outerset$ and $\lambda \in (0, 1)$, it holds that $\coupledset(\lambda \outer + (1-\lambda)\outer^\prime) \supseteq \lambda \coupledset(\outer) + (1-\lambda)\coupledset(\outer^\prime)$), when, $\constr[\numconstr]$ is quasi-concave in $(\outer, \inner)$ for all $\numconstr \in [\numconstrs]$ and $\innerset$ is convex, conditions that guarantee the \emph{concavity\/} of the constraint correspondence are more specialized.%
\footnote{There exists a wide range of conditions on $\constr$ that guarantee the concavity of $\coupledset$ (see Section 2 of \citet{nikodem1989concave} and Chapter 36 of \citet{czerwik2002functional}.}
In certain applications, such as homothetic Fisher markets (\Cref{sec:fisher}), $\obj$ is not convex in $(\outer, \inner)$; thus 4$'$a is not satisfied.
Thus, we also provide the following alternative set of assumptions in lieu of Part 4 of \Cref{main-assum}, which also guarantee that $\val$ is convex, and which homothetic Fisher markets satisfy.%
\footnote{A proof of this fact is provided in \Cref{sec-app:prelim}, \Cref{thm:convex-value-func}.}

\begin{assumption}[Alternative 2 to Part 4 of Assumption \ref{main-assum}]
\label{assum:alter_2}
4$'$a.~$\obj$ is convex in $\outer$, for all $\inner \in \innerset$; and
4$'$b.~the constraints $\constr$ are of the form $\constr[1](\outer, \inner[1]), \ldots,  \constr[\numconstrs](\outer, \inner[\numconstrs])$;%
\footnote{To simplify notation, we assume $\innerdim = \numconstrs$, but this restriction is unnecessary.}
4$'$c.~for all $\numconstr \in [\numconstrs]$, $\outer \in \outerset$, and $\inner \in \innerset$, $\frac{- \grad[{\inner[\numconstr]}]\obj(\outer, \inner)}{ \grad[{{\inner[\numconstr]}}] \constr[\numconstr](\outer, \inner[\numconstr])} \constr[\numconstr](\outer, \inner[\numconstr])$ is convex in $\outer$. 
\end{assumption}

\begin{remark}
The value function is guaranteed to be strongly convex under either \Cref{assum:alter_1} or \ref{assum:alter_2}, if we also assume that $\obj(\outer, \inner)$ is strongly convex in $\outer$, for all $\inner \in \innerset$.
\end{remark}

\section{First-Order Methods via an Envelope Theorem}
\label{sec:envelope}

The envelope theorems, popular tools in mathematical economics, allow for explicit formulas for the gradient of the value function in min-max games, even when the action sets are dependent.
\citet{afriat1971envelope} appears to have been the first make use of the Lagrangian to differentiate the value function, though his conclusion was later obtained under weaker assumptions by \citet{milgrom2002envelope}.
%

\citet{milgrom2002envelope}'s
envelope theorem
provides an explicit formula for the gradient of the value function. 
When action sets are independent, this function is guaranteed to be differentiable under mild assumptions \cite{nouiehed2019solving}.
When action sets are dependent, however, it is not necessarily differentiable, as seen in \Cref{no-danskins-ex}.
As a remedy, we present a subdifferential envelope theorem for non-differentiable but convex value functions.

\begin{theorem}[Subdifferential Envelope Theorem]
\label{envelope-sd}
Consider the value function $\val(\outer) = \max_{\inner \in \innerset: \constr(\outer, \inner) \geq 0} \obj(\outer, \inner) $.
Under \Cref{main-assum},
%
at any point $\widehat{\outer} \in \outerset$, $\subdiff[\outer] \val(\widehat{\outer}) =$
\begin{align}
\label{envelope-deriv-sd}
         \mathrm{conv} \left( \bigcup_{\inner^*(\widehat{\outer}) \in \innerset^*(\widehat{\outer})}      \bigcup_{\langmult[\numconstr]^*(\widehat{\outer}, \inner^*(\widehat{\outer})) \in \langmults[\numconstr]^*(\widehat{\outer}, \inner^*(\widehat{\outer}))} \left\{ \grad[\outer] \obj\left( \widehat{\outer}, \inner^{*}(\widehat{\outer})\right) + \sum_{\numconstr = 1}^{\numconstrs} \langmult[\numconstr]^*(\widehat{\outer}, \inner^*(\widehat{\outer})) \grad[\outer] \constr[k]\left(\widehat{\outer}, \inner^{*}(\widehat{\outer})\right) \right\}
    \right) \enspace ,
\end{align}
\noindent
where $\mathrm{conv}$ is the convex hull operator and $\langmult^*(\widehat{\outer}, \inner^*(\widehat{\outer})) = \left(\langmult[1]^*(\widehat{\outer}, \inner^*(\widehat{\outer})), \ldots, \langmult[\numconstrs]^*(\widehat{\outer}, \inner^*(\widehat{\outer})) \right)^T \in \langmults^*(\widehat{\outer}, \inner^*(\widehat{\outer}))$ are the optimal KKT multipliers associated with $\inner^{*}(\widehat{\outer}) \in \innerset^*(\widehat{\outer})$.
\end{theorem}

The envelope theorem states that the gradient of a differentiable value function is the gradient of the Lagrangian evaluated at the optimal solution.
Generalizing this fact,
our subdifferential envelope theorem states that every subgradient of the value function, $\val(\outer) = \max_{\inner \in \innerset: g(\outer, \inner) \geq 0} \obj(\outer, \inner)$ is a convex combination of the values of the gradient of the Lagrangian evaluated at the optimal solutions $(\inner^*(\outer), \langmult^*(\outer, \inner^*(\outer))) \in \innerset^*(\outer) \times \langmults^*(\outer, \inner^*(\outer))$.

With our envelope theorem in hand, we are now ready to present two gradient-descent/ascent-type algorithms for min-max Stackelberg games, which follow the gradient of the value function.%

Our first algorithm, \mydef{max-oracle gradient-descent}, following \citet{jin2020local}, assumes access to a max-oracle,
which given $\outer \in \outerset$, returns a $\delta$-best-response for the $\inner$ player.
That is, for all $\outer \in \outerset$, the max-oracle returns $\widehat{\inner} \in \innerset$ s.t.\ $\constr(\outer, \widehat{\inner}) \geq 0$ and $f(\outer, \widehat{\inner}) \geq \max_{\inner \in \innerset : \constr(\outer, \inner) \geq 0} f(\outer, \inner) - \delta$.
It then runs (sub)gradient descent on the outer player's value function,
using \Cref{envelope-sd} to compute the requisite subgradients.
Inspired by the multi-step gradient-descent algorithm of \citet{nouiehed2019solving} and \citeauthor{goodfellow2020generative}'s algorithm to train generative adversarial networks \cite{goodfellow2020generative}, our second algorithm, \mydef{nested gradient-descent/ascent} (\Cref{ngd}), computes both $\outer^{*}$ and $\inner^{*}$ explicitly, without oracle access.
We simply replace the max-oracle in our max-oracle gradient-descent algorithm by a projected gradient-ascent procedure, which again computes a $\delta$-best-response for the $\inner$ player.

Once $\widehat{\inner}$ is found at iteration $t$, one can compute
optimal KKT multipliers $\langmult[1]^*(\outer^{(\outeriter)}, \widehat{\inner}(\outer^{(\outeriter)}))$, $\ldots$, $\langmult[\numconstrs]^*(\outer^{(\outeriter)}, \widehat{\inner}(\outer^{(\outeriter)}))$ for the 
outer player's value function, either via a system of linear equations using the complementary slackness conditions and the value of the objective function at the optimal, namely $(\outer^{(t)}, \widehat{\inner} (\outer^{(t)}))$, or by running gradient descent on the Lagrangian for the dual variables.
Additionally, most algorithms solving convex programs will return $\langmult^*(\outer^{(\outeriter)}, \widehat{\inner}(\outer^{(\outeriter)})) = (\langmult[1]^*(\outer^{(\outeriter)}, \widehat{\inner}(\outer^{(\outeriter)})), \ldots, \langmult[\numconstrs]^*(\outer^{(\outeriter)}, \widehat{\inner}(\outer^{(\outeriter)})))$ together with the optimal $\widehat{\inner}(\outer^{(\outeriter)}))$ without incurring any additional computational expense.
As a result, we assume that the optimal KKT multipliers $\langmult^*(\outer^{(\outeriter)}, \widehat{\inner}(\outer^{(\outeriter)}))$ associated with a solution $\widehat{\inner}(\outer^{(\outeriter)}) \in \argmax_{\inner \in \innerset:  \constr(\outer^{(\outeriter)}, \inner) \geq \zeros} \obj(\outer^{(\outeriter)}, \inner)$ can be computed in constant time. 

Having explained our two procedures, our next task is to derive their convergence rates.
It turns out that under very mild assumptions, i.e., when \Cref{main-assum} holds, the outer player's value function is Lipschitz continuous in $\outer$.
More precisely, under \Cref{main-assum} the value function is $\lipschitz[\val]$-Lipschitz continuous, where $\lipschitz[\val] = \max_{(\widehat{\outer}, \widehat{\inner}) \in \outerset \times \innerset} \left\| \grad[\outer] \obj\left( \widehat{\outer}, \inner^{*}(\widehat{\outer})\right) + \sum_{\numconstr = 1}^{\numconstrs} \langmult[\numconstr]^*(\widehat{\outer}, \inner^*(\widehat{\outer})) \grad[\outer] \constr[k]\left(\widehat{\outer}, \inner^{*}(\widehat{\outer})\right)\right\|$.%
\footnote{This max norm is well-defined since $\grad[\outer] \obj, \grad[\outer] \constr[1], \hdots, \grad[\outer] \constr[\numconstrs]$ are continuous, the constraint set is non-empty and compact, and by Slater's condition, optimal KKT multipliers are guaranteed to exist. By \Cref{envelope-sd} the norm of all subgradients of the value function are bounded by $\lipschitz[\val]$, implying that $\val$ is $\lipschitz[\val]$-Lipschitz continuous. Additionally, under Slater's condition various upper bounds on the KKT multipliers are known (e.g., \citet{nedic2009approximate} or Chapter VII, Theorem 2.3.3, \citet{urruty1993convex}), which simplify the computation of $\lipschitz[\val]$ (since exact values are not needed).}
The Lipschitz-continuity of the value function in turn suggests that an $(\varepsilon, \varepsilon)$-Stackelberg equilibrium should be computable in $O(\varepsilon^{-2})$ 
iterations by our max-oracle gradient descent algorithm (\Cref{mogd}), since our method is a subgradient method. 

\begin{algorithm}[t!]
\caption{Max-Oracle Gradient Descent}
\label{mogd}
\textbf{Inputs:} $\outerset, \innerset, \obj, \constr, \learnrate, \outeriters, \outer^{(0)}$ \\ 
\textbf{Output:} $(\outer^{*}, \inner^{*})$
\begin{algorithmic}[1]
\For{$\outeriter = 1, \hdots, \outeriters + 1$}
    \State Find $\inner^{(\outeriter-1)} \in \innerset$ s.t.\ $\obj (\outer^{(\outeriter-1)}, \inner^{(\outeriter-1)}) \geq \val(\outer^{(\outeriter-1)})
    - \delta$ \& $\constr (\outer^{(\outeriter-1)}, \inner^{(\outeriter-1)}) \geq \zeros$
    \State Set $\langmult^{(\outeriter-1)} \gets \langmult^*(\outer^{(\outeriter-1)}, \inner^{(\outeriter-1)})$
    \State Set $\outer^{(\outeriter)} \gets \project[\outerset] \left( \outer^{(\outeriter-1)} - \learnrate[\outeriter] \left[\grad[\outer] \obj (\outer^{(\outeriter-1)}, \inner^{(\outeriter-1)}) + \sum_{\numconstr = 1}^\numconstrs \langmult[\numconstr]^{(\outeriter-1)} \grad[\outer]  \constr[\numconstr](\outer^{(\outeriter-1)}, \inner^{(\outeriter-1)}) \right] \right)$
\EndFor
\State \Return $(\outer^{(\outeriter)}, \inner^{(\outeriter)})_{\outeriter = 1}^\outeriters$
\end{algorithmic}
\end{algorithm}

\begin{theorem}
\label{l-mogd}
Consider a min-max Stackelberg game $(\outerset, \innerset, f, \constr)$ and suppose that \Cref{main-assum} holds. 
%
Suppose that \Cref{mogd} is run with step sizes $\{\learnrate[\outeriter]\}_{\outeriter} \subset \R_+$ s.t.\ $\sum_{k = 1}^{\infty}\learnrate[k] = \infty$ and $\sum_{k = 1}^{\infty}\learnrate[k][2] \leq \infty$, and outputs $(\outer^{(\outeriter)}, \inner^{(\outeriter)})_{\outeriter = 1}^\outeriters$. For any $\outeriter \in \N_{++}$, define the best iterate $(\bestiter[x][t], \bestiter[y][t]) \in \argmin_{(\outer^{(k)}, \inner^{(k)}) : k \in [t]} \obj(\outer^{(k)}, \inner^{(k)})$, the approximate subgradient $\subgrad (\outeriter) = \grad[\outer] \obj (\outer^{(\outeriter)}, \inner^{(\outeriter)}) + \sum_{\numconstr = 1}^\numconstrs \langmult[\numconstr]^*(\outer^{(\outeriter)}, \inner^{(\outeriter)}) \grad[\outer]  \constr[\numconstr](\outer^{(\outeriter)}, \inner^{(\outeriter)})$, and the gradient approximation error $\overline{\err[\outeriters]} = \frac{\sum_{\outeriter = 1}^{\outeriters} \learnrate[\outeriter] \left\|\subgrad (\outeriter) - \grad \val(\outer^{(\outeriter)})\right\| \left\|\outer^{(\outeriter)} - \outer^{*} \right\|}{\sum_{\outeriter = 1}^{\outeriters} \learnrate[\outeriter]}$.

Then, it holds that $\lim_{\outeriters \to \infty} f(\bestiter[x][\outeriters], \bestiter[x][\outeriters]) = f(\outer^{*}, \inner^{*}) + \overline{\err[\outeriters]}$ where for any $\varepsilon, \delta > 0$, let $(\outer^{*}, \inner^{*})$ be a $(\varepsilon, \delta)$-Stackelberg equilibrium.

Furthermore, for $\varepsilon \in (0,1)$, there exists an iteration $\outeriters \in O(\varepsilon^{-2})$, such that  for all $\outeriters^{*} \leq \outeriters$, $(\bestiter[x][\outeriters^{*}], \bestiter[y][\outeriters^{*}])$ is an $(\varepsilon + \overline{\err[\outeriters]}, \delta)$-Stackelberg equilibrium.
\end{theorem}

Note that the action profile \Cref{mogd} converges to depends on how well the output of the max-oracle allows us to approximate the gradient of the value function, which is captured by the term $\overline{\err[\outeriters]}$.
\sdeni{}{Closer to the Stackelberg equilibrium action $\outer^*$, i.e., for smaller $\left\|\outer^{(\outeriter)} - \outer^{*} \right\|$, the gradient approximation error, i.e., $\left\|\subgrad (\outeriter) - \grad \val(\outer^{(\outeriter)})\right\|$ matters less.
Note also that $\overline{\err[\outeriters]}$ can be further bounded as a function of the accuracy of the max-oracle $\delta$ if we assume in addition that $\obj$ and $\constr$ are Lipschitz-smooth and $\obj$ is bilipschitz, i.e., for some $\lipschitz[\obj] > 0$, $\forall \outer_1, \outer_2 \in X, \nicefrac{1}{\lipschitz[\obj]} \left\| \outer_1 - \outer_2 \right\| \leq \left\| \obj(\outer_1) - \obj(\outer_2) \right\| \leq \lipschitz[\obj] \left\| \outer_1 - \outer_2 \right\|$.%
\footnote{We note that Lipchitz-smoothness is a standard assumption in the optimization literature \cite{boyd2004convex}, and that bilipschitz continuity natural, as it implies that in addition to the gradient of the objective being bounded from above, i.e., Lipschitz-continuity, the norm of the gradient of the objective is also bounded away from zero, meaning that bilipschitz continuity captures all objectives whose solution occurs at a boundary of the constraints, i.e., the constraints are not vacuous.}
In particular, by the Lipschitz-smoothness of $\obj$ and $\constr$, the Lagrangian $\calL(\inner, \langmult; \outer) \doteq \obj(\outer, \inner) + \sum_{\numconstr = 1}^\numconstrs \langmult[\numconstr] \constr[\numconstr](\outer, \inner)$ is Lipschitz-smooth, in which case $\left\|\subgrad (\outeriter) - \grad \val(\outer^{(\outeriter)})\right\| \leq \lipschitz[\grad \calL]\|\inner^{(\outeriter)} - \inner^{*}(\outer)\|$, where $\lipschitz[\grad \calL] \in \R_+$ is the Lipschitz-smoothness coefficient of $\calL$.
Additionally, by the definition of the max-oracle and the bilipschitzness of $\obj$, it holds that $\delta \geq \obj(\outer^{(\outeriter)}, \inner^{*}(\outer^{(\outeriter)})) - \obj(\outer^{(\outeriter)}, \inner^{(\outeriter)}) \geq \nicefrac{1}{\lipschitz[\obj]}\|\inner^{(\outeriter)} - \inner^{*}(\outer)\|$.
Combining these two bounds yields $\left\|\subgrad (\outeriter) - \grad \val(\outer^{(\outeriter)})\right\| \leq \lipschitz[\grad \calL] \lipschitz[\obj] \delta$.
Finally, setting $c \doteq \max_{\x, \x^\prime \in \outerset} \left\| \x - \x^\prime \right\|$, we conclude that $c\lipschitz[\grad \calL] \lipschitz[\obj] \delta \geq |\overline{\err[\outeriters]}|$.}

\begin{algorithm}[htbp]
\caption{Nested Gradient Descent}
\label{ngd}
\textbf{Inputs:} $\outerset, \innerset, \obj, \constr, {\learnrate}^\outer, {\learnrate}^\inner, \outeriters_\outer, \outeriters_\inner,  \outer^{(0)}, \inner^{(0)}$ \\ 
\textbf{Output:} $\outer^{*}, \inner^{*}$
\begin{algorithmic}[1]
\For{$\outeriter = 1, \hdots, \outeriters_\outer + 1$}
    \State  $\inner^{(\outeriter-1)} = \inner^{(0)}$
    \For{$s = 1, \hdots, \outeriters_\inner$}
        \State $\inner^{(\outeriter-1)} = \project[{\{\inner \in \innerset: \constr (\outer^{(\outeriter-1)}, \inner) \geq \zeros\}}] \left( \inner^{(\outeriter-1)} + \learnrate[s][\inner] \left[ \grad[\inner] \obj (\outer^{(\outeriter-1)}, \inner^{(\outeriter-1)}) \right] \right)$
    \EndFor
    \State Set $\langmult^{(\outeriter-1)} = \langmult^*(\outer^{(\outeriter-1)}, \inner^{(\outeriter-1)})$
    \State Set $\outer^{(\outeriter)} = \project[\outerset] \left( \outer^{(\outeriter-1)} - \learnrate[\outeriter][\outer] \left[\grad[\outer] \obj (\outer^{(\outeriter-1)}, \inner^{(\outeriter-1)}) + \sum_{\numconstr = 1}^\numconstrs \langmult[\numconstr]^{(\outeriter-1)} \grad[\outer]  \constr[\numconstr](\outer^{(\outeriter-1)}, \inner^{(\outeriter-1)}) \right] \right)$
\EndFor
\State \Return $(\outer^{(\outeriter)}, \inner^{(\outeriter)})_{\outeriter = 1}^{\outeriters_\outer}$
\end{algorithmic}
\end{algorithm}

As is expected, the $O(\varepsilon^{-2})$ iteration complexity can be improved to $O({\varepsilon}^{-1})$, if additionally, $\val$ is strongly convex in $\outer$.
(See \Cref{sec-app:algos}, \Cref{ls-mogd}). Combining the convergence results for our max-oracle gradient descent algorithm with convergence results for gradient descent \cite{boyd2004convex}, we obtain the following convergence rates for the nested gradient descent-ascent algorithm (\Cref{ngd}).
We include the formal proof and statement for the case when \Cref{main-assum} holds and $\obj$ is Lipschitz-smooth in \Cref{sec-app:algos} (\Cref{l-ngd}).
The other results follow similarly.
\begin{theorem}
\label{ngd-thms}
Consider a min-max Stackelberg game, $(\outerset, \innerset, \obj, \constr)$ and suppose
that \Cref{main-assum} holds.
Then, under standard assumptions on the step sizes, the iteration complexities given 
below hold for the computation of a $(\varepsilon, \varepsilon)$-Stackelberg equilibrium:
%
\renewcommand*\arraystretch{1.33}
\begin{center}
\begin{tabular}{|l||l|l|}\hline
 & $\obj$ is $\lipschitz[\grad \obj]$-smooth & $\obj$ is $\lipschitz[\grad \obj]$-smooth \\ & & + $\obj$ strongly concave in $\inner$  \\ \hline \hline
\Cref{main-assum} & $O(\varepsilon^{-3})$ &  $O\left(\varepsilon^{-2} \log(\varepsilon^{-1}) \right)$ \\\hline
\Cref{main-assum} + $\val$ strongly convex in $\outer$ & $O(\varepsilon^{-2})$ & $O(\varepsilon^{-1} \log(\varepsilon^{-1}))$ \\ \hline
\end{tabular}
\end{center}
\renewcommand*\arraystretch{1}
\end{theorem}


Since the value function in the convex-concave dependent setting is not guaranteed to be differentiable (see \Cref{no-danskins-ex}), we cannot ensure that the objective function is Lipschitz-smooth in general.
Thus, unlike previous results for the independent setting that required this latter assumption to achieve faster convergence (e.g., \cite{nouiehed2019solving}), in our analysis of \Cref{mogd}, we assume only that the objective function is continuously differentiable, which leads to a more widely applicable, albeit slower, convergence rate.
Note, however, that we assume Lipschitz-smoothness in our analysis of \Cref{ngd}, as it allows for faster convergence to the $\inner$-player's optimal strategy, but this assumption could also be done away with again, at the cost of a slower convergence rate.

\section{An Economic Application: Fisher Markets}
\label{sec:fisher}

The Fisher market model, attributed to Irving Fisher \cite{brainard2000compute}, has received a great deal of attention recently, in particular by computer scientists, as its applications to fair division and mechanism design have proven useful for the design of automated markets in many online marketplaces.
In this section, we argue that 
a competitive equilibrium in Fisher markets can be understood 
a Stackelberg equilibrium of a convex-concave min-max Stackelberg game.
We then apply our first-order methods to compute these equilibria in various Fisher markets.


A \mydef{Fisher market} consists of $\numbuyers$ buyers and $\numgoods$ divisible goods \cite{brainard2000compute}. Each buyer $\buyer \in \buyers$ has a budget $\budget[\buyer] \in \mathbb{R}_{+}$, \sdeni{}{a \mydef{consumption set} $\consumptions[\buyer] \subset \R^\numgoods_+$,} and a \mydef{utility function} $\util[\buyer]: \mathbb{R}_{+}^{\numgoods} \to \mathbb{R}_+$. \sdeni{}{We also define the \mydef{space of joint consumption}, i.e., $\consumptions \doteq \bigtimes_{\buyer \in \buyers} \consumptions[\buyer] \subset \R^{\numbuyers \times \numgoods}_+$.} 
As is standard in the literature, we assume that there is one divisible unit of each good 
in the market \cite{AGT-book}. An instance of a Fisher market is given by a tuple $(\numbuyers, \numgoods, \consumptions, \util, \budget)$, where $\util = \left\{\util[1], \hdots, \util[\numbuyers] \right\}$ is a set of utility functions, one per buyer, and $\budget \in \R_{+}^{\numbuyers}$ is the vector of buyer budgets, for which, without loss of generality, we assume $\sum_{\buyer \in \buyers} \budget[\buyer] = 1$.
We abbreviate a Fisher market as $(\consumptions, \util, \budget)$ when $\numbuyers$ and $\numgoods$ are clear from context. 

A function $f: \R^m \to \R$ is said to be \mydef{homogeneous of degree $k$} if $\forall \allocation[ ] \in \R^m, \lambda > 0, f(\lambda \allocation[ ]) = \lambda^k f(\allocation[ ])$.
A Fisher market $(\consumptions, \util, \budget)$ is said to be \mydef{homothetic} if, for all buyers $\buyer \in \buyers$, $\util[\buyer]$ is a continuous 
and homogeneous of degree 1, i.e., for all $\lambda \in \R_+$ $\util[\buyer](\lambda \allocation[\buyer]) = \lambda \util[\buyer](\allocation[\buyer])$.

Goods are assigned \mydef{prices} $\price = \left(\price[1], \hdots, \price[\numgoods] \right)^T \in \mathbb{R}_+^{\numgoods}$. An \mydef{allocation} $\allocation = \left(\allocation[1], \hdots, \allocation[\numbuyers] \right)^T \in \R_+^{\numbuyers \times \numgoods}$ is a map from goods to buyers, represented as a matrix, s.t. $\allocation[\buyer][\good] \ge 0$ denotes the amount of good $\good \in \goods$ allocated to buyer $\buyer \in \buyers$. A tuple $(\price^*, \allocation^*)$ is said to be a \mydef{competitive (or Walrasian) equilibrium} of Fisher market $(\consumptions, \util, \budget)$ if 1.~buyers are utility maximizing, constrained by their budget, i.e., $\forall \buyer \in \buyers, \allocation[\buyer]^* \in \argmax_{\allocation[ ] \in \consumptions[\buyer]: \allocation[ ] \cdot \price^* \leq \budget[\buyer]} \util[\buyer](\allocation[ ])$;
and 2.~the market clears, i.e., $\forall \good \in \goods,  \price[\good]^* > 0 \Rightarrow \sum_{\buyer \in \buyers} \allocation[\buyer][\good]^* = 1$ and $\price[\good]^* = 0 \Rightarrow\sum_{\buyer \in \buyers} \allocation[\buyer][\good]^* \leq 1$.

We now formulate the problem of computing a competitive equilibrium $(\price^*, \allocation^*)$ of a Fisher market $(\consumptions, \util, \budget)$, where $\util$ is a set of continuous, concave, and homogeneous utility functions, as a convex-concave min-max Stackelberg game,
a perspective which
has not been taken before. 
Fisher markets can by solved via the Eisenberg-Gale convex program~\cite{eisenberg1959consensus}.
Recently, \citet{cole2019balancing} derived a convex program, which differs from the dual of the Eisenberg-Gale program by a constant factor \cite{goktas2020cchwine}, namely:
\begin{align}
    \min_{\price \in \simplex[\numgoods]} \sum_{\good \in \goods} \price[\good] + \sum_{\buyer \in \buyers} \budget[\buyer] \log \left( \max_{\allocation[\buyer] \in \consumptions[\buyer] : \allocation[\buyer] \cdot \price \leq \budget[\buyer]} \util[\buyer](\allocation[\buyer]) \right) \enspace .
    \label{cole-tao}
\end{align}

\noindent

Rearranging, we obtain the following convex-concave min-max Stackelberg game:
\begin{align}
    \min_{\price \in \simplex[\numgoods]} \max_{\allocation \in \consumptions :  \allocation \price \leq \budget} \sum_{\good \in \goods} \price[\good] + \sum_{\buyer \in \buyers}  \budget[\buyer] \log \left(  \util[\buyer](\allocation[\buyer]) \right) \enspace .
    \label{fisher-program}
\end{align}

This min-max game is played by a fictitious (Walrasian) auctioneer and a set of buyers, who effectively play as a team.
The objective function in this game is then the sum of the auctioneer's welfare (i.e., the sum of the prices) and the Nash social welfare of buyers (i.e., the second summation).
As the buyer's action set 
is dependent on the price vector $\price$ selected by the auctioneer, we cannot use existing first-order methods to solve this problem.
However, we can use Algorithms \ref{mogd} and \ref{ngd}.

Starting from \Cref{fisher-program}, define the auctioneer's value function $\val(\price) = \max_{\allocation \in \consumptions[\buyer]: \allocation \price \leq \budget} \sum_{\good \in \goods} \price[\good] + \sum_{\buyer \in \buyers}  \budget[\buyer] \log \left( \util[\buyer] (\allocation[\buyer]) \right)$, and buyer $\buyer$'s demand set  $\outerset_{\buyer}^*(\price, \budget) = \argmax_{\allocation[\buyer] \in \consumptions[\buyer]: \allocation[\buyer] \price \leq \budget} \util[\buyer] (\allocation[\buyer])$.
\Cref{envelope-sd} then provides the relevant subgradients so that we can run Algorithms~\ref{mogd} and \ref{ngd}, namely
$\subdiff[\price] \val(\price) = \ones[\numgoods] - \sum_{\buyer \in \buyers}  \outerset_{\buyer}^*(\price, \budget)$ and $\grad[{\allocation[\buyer]}]  \left( \sum_{\good \in \goods} \price[\good] + \sum_{\buyer \in \buyers}  \budget[\buyer] \log \left(  \util[\buyer](\allocation[\buyer])\right)\right) = \frac{\budget[\buyer]}{\util[\buyer](\allocation[\buyer])} \grad[{\allocation[\buyer]}] \util[\buyer](\allocation[\buyer])$,
using the Minkowski sum to add set-valued quantities, where $\ones[\numgoods]$ is the vector of ones of size $\numgoods$.%
\footnote{We include detailed descriptions of the algorithms applied to Fisher markets in \Cref{sec-app:fisher}.}

\citet{fisher-tatonnement} observed that solving the dual of the Eisenberg-Gale program (Equation~\ref{cole-tao}) via (sub)gradient descent \cite{devanur2002market} is equivalent to solving for a competitive equilibrium in a Fisher market using an auction-like economic price adjustment process named \emph{t\^atonnement} that was first proposed by L\'eon Walras in the 1870s \cite{walras}.
The \emph{t\^atonnement} process increases the prices of goods that are overdemanded and decreases the prices of goods that are underdemanded. 
Mathematically, the (vanilla) \emph{t\^atonnement} process \cite{arrow-hurwicz, walras} is defined as
$\price(t) = \max \left\{ \price(t-1) + \learnrate[t] \left(  \sum_{\buyer \in \buyers} \outer_{\buyer}^*(\price(t), \budget) - 1\right), 0 \right\}$ for $\price(0) \in \R^\numgoods_{+}$, 
where $\outer_{\buyer}^*(\price(t), \budget[\buyer]) \in  \outerset_{\buyer}^*(\price(t), \budget[\buyer]) = \argmax_{\allocation[\buyer] \geq \zeros : \allocation[\buyer] \price(t) \leq \budget[\buyer]} \util[\buyer](\allocation[\buyer]) $ is the demand set of buyer $\buyer$.
The max-oracle algorithm applied to \Cref{fisher-program} 
is then equivalent to a \emph{t\^atonnemement} process where the buyers report a $\delta$-utility maximizing demand.
Further, we have the following corollary of \Cref{l-mogd}.
%
\begin{corollary}
Let $(\consumptions, \util, \budget)$ be a Fisher market with equilibrium price vector $\price^{*}$, where $\util$ is a set of continuous, concave, homogeneous, and continuously differentiable utility functions\sdeni{}{, and the joint consumption space $\consumptions$ is bounded away from $\zeros$}. Consider the t\^atonnement process
~\cite{goktas2020cchwine}.%
Assume that the step sizes $\learnrate[\outeriter]$ satisfy the usual conditions: $ \sum_{k = 1}^{\outeriters}\learnrate[k][2] \leq \infty$ and $\sum_{k = 1}^{\outeriters}\learnrate[k] = \infty$. If $\bestiter[p][t] \in \argmin_{\price^{(k)}: k \in [t]} \val(\price^{(k)})$, then $\lim_{k \to \infty} \val(\bestiter[p][k]) = \val(\price^{*})$.
Additionally, t\^atonnement converges to an $\varepsilon$-competitive equilibrium in $O(\varepsilon^{-2})$ iterations.
\end{corollary}

If we also apply the nested gradient-descent-ascent algorithm to \Cref{fisher-program}, we arrive at an algorithm that is arguably more descriptive of market dynamics than \emph{t\^atonnement} itself, as it also includes the demand-side market dynamics of buyers optimizing their demands, potentially in a decentralized manner.
The nested \emph{t\^atonnement} algorithm essentially describes a two-step trial-and-error (i.e., \emph{t\^atonnement}) process, where first the buyers try to discover their optimal demand by increasing their demand for goods in proportion to the marginal utility the goods provide, and then the seller/auctioneer adjusts market prices by decreasing the prices of goods that are underdemanded and increasing the prices of goods that are overdemanded.
As buyers can calculate their demands in a decentralized fashion, the nested \emph{t\^atonnement} algorithm offers a more complete picture of market dynamics then the classic \emph{t\^atonnement} process.

\subsection{Experiments}
In order to better understand, the iteration complexity of Algorithms \ref{mogd} and \ref{ngd} (\Cref{sec-app:algos}), we ran a series of experiments on Fisher markets with three different classes of utility functions.%
\footnote{Our code can be found at \coderepo.}
Each utility structure endows \Cref{fisher-program} with different smoothness properties, which allows us to compare the efficiency of the algorithms under varying conditions.

Let $\valuation[\buyer] \in \R^\numgoods$, be a vector of parameters for the utility function of buyer $\buyer \in \buyers$.
We have the following utility function classes:
Linear: $\util[\buyer](\allocation[\buyer]) = \sum_{\good \in \goods} \valuation[\buyer][\good] \allocation[\buyer][\good]$,  Cobb-Douglas:  $\util[\buyer](\allocation[\buyer]) = \prod_{\good \in \goods} \allocation[\buyer][\good]^{\valuation[\buyer][\good]}$, Leontief:  $\util[\buyer](\allocation[\buyer]) = \min_{\good \in \goods} \left\{ \frac{\allocation[\buyer][\good]}{\valuation[\buyer][\good]}\right\}$.
\Cref{fisher-program} satisfies the smoothness properties listed in Table~\ref{tab:util-prop} when $\util$ is one of these three classes.
Our goals are two-fold.
First, we want to understand how the empirical convergence rates of Algorithms \ref{mogd} and \ref{ngd} (which, when applied to \Cref{fisher-program} give rise to Algorithms \ref{mogd:fm} and \ref{ngd:fm} in \Cref{sec-app:fisher}, respectively) compare to their theoretical guarantees under different utility structures.
Second, we want to understand the extent to which the convergence rates of these two algorithms differ in practice. We include a more detailed description of our experimental setup in \Cref{sec-app:fisher}.
%
\begin{table}[H]
\begin{center}
\begin{tabular}{|l|l|l|}\hline
 & $\val$ is differentiable & \Cref{main-assum} holds  \\ \hline
Linear & $\times$ & \checkmark \\ \hline
Cobb-Douglas & \checkmark  & \checkmark \\ \hline
Leontief & \checkmark & $\times$  \\\hline
\end{tabular}
\end{center}
\caption{\label{tab:util-prop}
Smoothness properties satisfied by \Cref{fisher-program} assuming different utility functions. Note that \Cref{main-assum} does not hold for Leontief utilities, because they are not differentiable.}
\end{table}  
\Cref{fig:experiments1} describes the empirical convergence rates of Algorithms \ref{mogd} and \ref{ngd} for linear, Cobb-Douglas, and Leontief utilities.
We observe that convergence is fastest in Fisher markets with Cobb-Douglas utilities, followed by linear, and then Leontief. 
We seem to obtain a tight convergence rate of $O(\nicefrac{1}{\sqrt{T}})$ for linear utilities, which seems plausible, as the value function is not differentiable assuming linear utilities, and hence we are unlikely to achieve a better convergence rate.
On the other hand, for Cobb-Douglas utilities, both the value and the objective function are differentiable; in fact, they are both twice continuously differentiable, making them both Lipschitz-smooth.
These factors combined seem to provide a much faster convergence rate than $O(\nicefrac{1}{\sqrt{T}})$.

Fisher markets with Leontief utilities, in which the objective function is not differentiable, are the hardest markets of the three for our algorithms to solve.
Indeed, our theory does not even predict convergence.
Still, convergence is not entirely surprising,
as \citet{fisher-tatonnement} have shown that buyer demand throughout \emph{t\^atonnement} is bounded for Leontief utilities, which means that the objective function of \Cref{fisher-program} is locally Lipschitz around t\^atonnement trajectories: i.e., any subgradient computed by the algorithm will be bounded.
Overall, our theory suggests that differentiability of the value function is not essential to guarantee convergence of first-order methods in convex-concave games, while our experiments seem to suggest that differentiability of the objective function is more important than differentiability of the value function in regards to the convergence rate.


\begin{figure}
    \centering
    \includegraphics[width=0.675\linewidth]{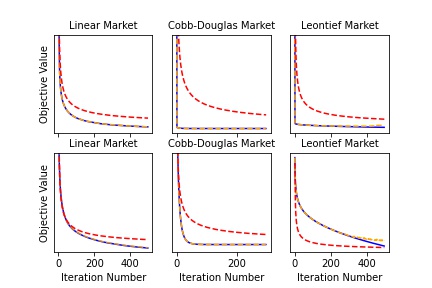}
    \caption{The first row describes the average trajectory of the value of the objective function for a randomly initialized market on each iteration of both \Cref{mogd:fm} (in {\color{blue} blue}) and \Cref{ngd:fm} (in {\color{orange} orange}) when the starting prices are high, while the second row describes the average trajectory of the objective function when starting prices are low for linear, Cobb-Douglas, and Leontief Fisher markets respectively. The dashed red line represents a convergence rate of $O(\nicefrac{1}{\sqrt{T}})$, which corresponds to an iteration complexity of $O(\nicefrac{1}{\epsilon^2})$.}
    \label{fig:experiments1}
\end{figure}
In order to investigate whether the outputs of \Cref{mogd}, which uses an exact max-oracle (i.e., 0-max-oracle) are more precise than those of \Cref{ngd}, we solved 500 randomly initialized markets with both algorithms.
We then ran a James' first-order test \cite{james1994test,stestsR} on the mean output of both algorithms to see if their difference was statistically significant.
Our calculations produced $p$-values of 
$0.69$, $0$, and $1.06 \times 10^{-18}$,
for Fisher markets with linear, Cobb-Douglas, and Leontief utilities, respectively.
At a significance level of $0.05$, these results are \emph{not\/} statistically significant for linear utilities only.
This result can be attributed to the fact that the value function is not differentiable in the linear case, which makes the nested gradient descent/ascent algorithm less precise.

\section{Conclusion}





In this paper, we study a class of constrained convex-concave min-max optimization problems with dependent constraint sets, which we call convex-concave min-max Stackelberg games.
As such games do not afford a minimax theorem in general, we focused on existence and computation of their Stackelberg equilibria.
We established the existence of Stackelberg equilibria in these games, assuming continuous objective functions and suitable convexity assumptions on the players' action sets.
We then introduced a novel subdifferential envelope theorem, which formed the core of two subgradient methods with polynomial-time iteration complexity that converge to Stackelberg equilibria.
Finally, we applied our theory to the computation of competitive equilibria in Fisher markets.
This application yielded a new variant of the classic \emph{t\^atonnement\/} process, where, in addition to the auctioneer iteratively adjusting prices, the buyers iteratively compute their demands.
A further question of interest both for Fisher market dynamics and convex-concave min-max Stackelberg games more generally is whether gradient-descent-ascent (GDA) converges in the dependent action set setting as it does in the independent action setting \cite{lin2020gradient}.
GDA dynamics for Fisher markets correspond to myopic best-response dynamics (see, for example, \citet{monderer1996potential}).
We would expect such a result to be of computational as well as economic interest.

\section*{Acknowledgments}

We would like to thank Dustin Morrill for his feedback on an early draft of this paper, and Niao He for pointing us to relevant literature.
This work was partially supported by NSF Grant CMMI-1761546.

\bibliographystyle{plainnat}
\bibliography{references}

\newpage
\appendix
\section{Pseudo-Games and Generalized Nash Equilibria}
\label{sec-app:GNE}

A two-player, zero-sum \mydef{pseudo-game} (or \mydef{abstract economy~\cite{arrow-debreu})}%
\footnote{We refer the reader to \citeauthor{facchinei2010generalized}'s (\citeyear{facchinei2007generalized, facchinei2010generalized}) survey on pseudo-games for a more detailed exposition, beyond two-player zero-sum pseudo-games.}
$(\outerset, \innerset, \obj, \constr)$ comprises two players, $\outer$ and $\inner$, with respective payoff functions $-\obj(\outer, \inner)$ and $\obj(\outer, \inner)$, and respective action spaces given by the correspondences $\outers: \innerset \rightrightarrows \outerset$ and $\inners: \outerset \rightrightarrows \innerset$, i.e., set-valued mappings, each of which depends on the choice the other player takes: i.e., $\outers(\inner) = \left\{\outer \in \outerset \mid 
\constr (\outer, \inner) \geq \bm{0} \right\}$ and $\inners(\outer) = \left\{\inner \in \innerset \mid 
\constr (\outer, \inner) \geq \bm{0} \right\}$.
%
%
A \mydef{generalized Nash equilibrium (GNE)}, the canonical solution concept in pseudo-games, 
is an action profile $(\outer^*, \inner^*) \in \outerset \times \innerset$ s.t.\ $\constr (\outer^*, \inner^*) \geq 0$ s.t.
\begin{align}\label{eq:gne_definition}
    \max_{\inner \in \innerset: \constr(\outer^*, \inner) \geq 0} \obj(\outer^*, \inner)  \leq \obj (\outer^*, \inner^*) \leq \min_{\outer \in \outerset : \constr (\outer, \inner^*) \geq 0} \obj(\outer, \inner^*) \enspace .
\end{align}

\noindent
That is, at a GNE, players choose best responses to the other players' strategies from within the space of strategies defined by the other players' choices.

It is difficult to imagine a situation in which the players choose some $\outer$ and $\inner$ simultaneously and then, for some reason, it happens that the constraints $\constr (\outer, \inner) \geq \zeros$ are satisfied \cite{ichiishi1983gne}.
For this reason, a pseudo-game is not technically a game.
If the players move sequentially, however, and only the inner player's feasible set is constrained by the outer player's choice (but not vice versa), then the pseudo-game is indeed a game---a Stackelberg game, to be precise.


When $\obj$ is convex-concave, and $\constr[\numconstr]$ is concave for all $\numconstr \in [\numconstrs]$, a GNE is guaranteed to exist \cite{arrow-debreu}. 
The existence of a GNE, however, does not imply that a minimax theorem holds, which in turn means that Stackelberg equilibria of a pseudo-game need not coincide with its GNE:

%

\begin{example}
\label{ex:gne_not_se}
Consider the constrained min-max optimization problem $\min_{x \in [-1,1]} \max_{y \in [-1,1] : x + y \leq 0} x^2 + y + 1$ 
with optimum $x^* = \nicefrac{1}{2}, y^* = -\nicefrac{1}{2}$ and value $\nicefrac{3}{4}$, as in Example~\ref{min-max-fail-ex}. Now, consider the same problem (i.e., the same objective function and constraints), with the order of the $\min$ and the $\max$ reversed:
$\max_{\inner[ ] \in [-1, 1]} \min_{\outer[ ] \in [-1, 1] : \outer[ ] + \inner[ ] \leq 0} \outer[ ]^2 + \inner[ ] + 1$.
The optimum is now $\outer[ ]^{**} = -1, \inner[ ]^{**} = 1$ with value $3$.
The min-max optimum $(x^*, y^*)$ is not a GNE, because $x^* = \nicefrac{1}{2}$ is not a best response to $y^* = -\nicefrac{1}{2}$ over the set $\{ x \in [-1,1] \mid x+y^* \le 0 \}$,  e.g., the $\outer$-player can do better by playing $x \in [-1, \nicefrac{1}{2})$.
However, the max-min optimum $(x^{**}, y^{**})$ is a GNE, because $y^{**} = 1$ is a best response to $x^{**} = -1$ over the set $\{ y \in [-1,1] \}$. 
In fact, this game has a set of GNEs $(x^\prime, y^\prime)$ given by $x^\prime \in [-1, 0]$, and $y^\prime = -x^\prime$, with values in the set $[1, 3]$
.
\end{example}


Because a minimax theorem does not hold for pseudo-games, the solutions to min-max Stackelberg games (i.e., Stackelberg equilibria)
do not necessarily coincide with the GNE of the associated pseudo-game.
Interestingly, the min-max (resp.\ max-min) value of the pseudo-game lower (resp.\ upper) bounds the value of the pseudo-game at any GNE.
In other words, the payoff of the $\outer$ (resp.\ $\inner$) player at a GNE is no lower (resp.\ higher) than their payoff at any Stackelberg equilibrium: i.e.,
%
%
\begin{align}
    \min_{\outer \in \outerset } \max_{\inner \in \innerset: \constr(\outer, \inner) \geq 0} \obj(\outer, \inner) \leq \obj(\outer^*, \inner^*) \leq   \max_{\inner \in \innerset} \min_{\outer \in \outerset : \constr(\outer, \inner) \geq 0} \obj(\outer, \inner)
\end{align}

This can be observed by taking the minimum over all $\outer^* \in \outerset$ on the left hand side of \Cref{eq:gne_definition}:
\begin{align}
    \min_{\outer \in \outerset } \max_{\inner \in \innerset: \constr(\outer, \inner) \geq 0} \obj(\outer, \inner) \leq  \max_{\inner \in \innerset: \constr(\outer^*, \inner) \geq 0} \obj(\outer^*, \inner)  \leq  
    \obj(\outer^*, \inner^*)
    \enspace ,
\end{align}
and the maximum over all $\inner^* \in \innerset$ on the right hand side of \Cref{eq:gne_definition}
\begin{align}
    \obj (\outer^*, \inner^*) \leq \min_{\outer \in \outerset : \constr (\outer, \inner^*) \geq 0} \obj(\outer, \inner^*) \leq \max_{\inner \in \innerset} \min_{\outer \in \outerset : \constr (\outer, \inner) \geq 0} \obj(\outer, \inner)
    \enspace .
\end{align}

\section{Omitted Proofs \Cref{sec:prelim}}
\label{sec-app:prelim}

\stacklebergexistence*
\begin{proof}[Proof of \Cref{stackleberg-existence}]
By Theorem 5.9 of \citet{rockafellar2009variational}, the constraint correspondence $C(\outer) = \left\{\inner \in \innerset \mid \constr(\outer, \inner) \geq \zeros \right\}$ is continuous, i.e., upper and lower hemi-continuous (see for instance chapter 5 of \citet{rockafellar2009variational}).
Hence, by the Maximum Theorem \cite{berge1997topological}, the outer player's value function $\val(\outer) = \max_{\inner \in \innerset: \constr(\outer, \inner) \geq \zeros} \obj(\outer, \inner) = \max_{\inner \in C(\outer)} \obj(\outer, \inner)$ is continuous, and the inner player's solution correspondence $\innerset^*(\outer) = \argmax_{\inner \in \innerset: \constr(\outer, \inner) \geq \zeros} \obj(\outer, \inner)$ is non-empty, for all $\outer \in \outerset$.
Since $\val$ is continuous and $\outerset$ is compact and non-empty, by the extreme value theorem \cite{protter2012extreme}, there exists a minimizer $\outer^* \in \outerset$ of $\val$.
Hence $(\outer^*, \inner^*(\outer^*))$, with $\inner^*(\outer^*) \in \innerset^*(\outer^*)$, is a Stackelberg equilibrium of $(\outerset, \innerset, \obj, \constr)$.

Let $(\outer_1, \inner_1)$ and $(\outer_2, \inner_2)$ be two Stackelberg equilibria whose values differ.
WLOG, suppose $\obj(\outer_1, \inner_1) > \obj(\outer_2, \inner_2)$, so that $\val(\outer_1) = \obj(\outer_1, \inner_1) > \obj(\outer_2, \inner_2) = \val(\outer_2)$, where the first and last equality follow from the definition of Stackelberg equilibrium.
But then $(\outer_1, \inner_1)$ cannot be a Stackelberg equilibrium, since $\outer_1$ is not a minimizer of the outer player's value function.
Therefore, there cannot exist two Stackelberg equilibria whose values differ, i.e., the value of all Stackelberg equilibria is unique.
\end{proof}

\begin{proposition}\label{thm:convex-value-func}
Consider a min-max Stackelberg game $(\outerset, \innerset, \obj, \constr)$.
Suppose that
1.~(Slater's condition) 
$\forall \outer \in \outerset, \exists \widehat{\inner} \in \innerset$ s.t.\ $g_{\numconstr}(\outer, \widehat{\inner}) > 0$, for all $\numconstr \in [\numconstrs]$.
2.~$\grad[\outer] f, \grad[\outer] \constr[1], \ldots, \grad[\outer] \constr[\numconstrs]$ are continuous;
3$'$a.~$\obj$ is continuous and convex-concave,
3$'$b.~the constraints $\constr$ are of the form $\constr[1](\outer, \inner[1]), \constr[2](\outer, \inner[2]), \ldots, \constr[\numconstrs](\outer, \inner[\numconstrs])$%
\footnote{To simplify notation, we assume $\innerdim = \numconstrs$, but this theorem holds in general.}
are continuous in $(\outer, \inner)$, and concave in $\inner$, for all $\outer \in \outerset$;
3$'$c.~for all $\numconstr \in [\numconstrs]$, $\outer \in \outerset$, and $\inner \in \innerset$, $\frac{- \grad[{\inner[\numconstr]}]\obj(\outer, \inner)}{ \grad[{{\inner[\numconstr]}}] \constr[\numconstr](\outer, \inner[\numconstr])} \constr[\numconstr](\outer, \inner[\numconstr])$ is convex in $\outer$.
Then the value function $\val$ associated with $(\outerset, \innerset, \obj, \constr)$ is convex.
\end{proposition}

\begin{proof}[Proof of \Cref{thm:convex-value-func}]
By the Maximum Theorem \cite{berge1997topological}, the outer player's value function 
is continuous.

Define $\lang: \innerset \times \R_+^\numconstrs \times \outerset \to \R$ s.t.\ $\lang(\inner, \langmult; \outer) = \obj(\outer, \inner) + \sum_{\numconstr = 1}^\numconstrs \langmult[\numconstr] \constr[\numconstr](\outer, \inner)$ 
.
Since Slater's condition is satisfied, the KKT theorem \cite{kuhn1951kkt} applies, which means that for all $\outer \in \outerset$ and $\inner \in \innerset$, the optimal KKT multipliers $\langmult^*$ exist, and thus:
\begin{align}
    \grad[\inner] \obj(\outer, \inner) + \sum_{\numconstr = 1}^\numconstrs \langmult[\numconstr]^* \grad[\inner] \constr[\numconstr](\outer, \inner) = \zeros\\
    \grad[\inner] \obj(\outer, \inner) + \sum_{\numconstr = 1}^\numconstrs \langmult[\numconstr]^* \grad[\inner] \constr[\numconstr](\outer, \inner[\numconstr]) = \zeros && \text{(Assumption 3$'$b)} \\
    \grad[{\inner[\numconstr]}] \obj(\outer, \inner) + \sum_{\numconstr = 1}^\numconstrs \langmult[\numconstr]^* \grad[{\inner[\numconstr]}] \constr[\numconstr](\outer, \inner[\numconstr]) = \zeros && \\
    \langmult[\numconstr]^* = \frac{- \grad[{\inner[\numconstr]}] \obj(\outer, \inner)}{\grad[{\inner[\numconstr]}] \constr[\numconstr](\outer, \inner[\numconstr])} && \forall \numconstr \in [\numconstrs]
\end{align}

Plugging the optimal KKT multipliers back into the Lagrangian, we obtain:
\begin{align}
    \val(\outer) = \max_{\inner \in \innerset} \lang(\inner, \langmult^*(\outer, \inner); \outer) = \max_{\inner \in \innerset} \obj(\outer, \inner) + \sum_{\numconstr = 1}^\numconstrs \frac{- \grad[{\inner[\numconstr]}] \obj(\outer, \inner)}{\grad[{\inner[\numconstr]}] \constr[\numconstr](\outer, \inner[\numconstr])} \constr[\numconstr](\outer, \inner[\numconstr])
\end{align}

By Assumption 3$'$c, for all $\numconstr \in [\numconstrs]$, $\frac{- \grad[{\inner[\numconstr]}] \obj(\outer, \inner)}{\grad[{\inner[\numconstr]}] \constr[\numconstr](\outer, \inner[\numconstr])} \constr[\numconstr](\outer, \inner[\numconstr])$, and by Assumption 3$'$a, $\obj(\outer, \inner)$ are convex in $\outer$, for all $\inner \in \innerset$.
Therefore, $\obj(\outer, \inner) + \sum_{\numconstr = 1}^\numconstrs \frac{- \grad[{\inner[\numconstr]}] \obj(\outer, \inner)}{\grad[{\inner[\numconstr]}] \constr[\numconstr](\outer, \inner[\numconstr])} \constr[\numconstr](\outer, \inner[\numconstr])$ is convex in $\outer \in \outerset$, for all $\inner \in \innerset$.
Finally, by Danskin's theorem, $\val(\outer) = \max_{\inner \in \innerset} \obj(\outer, \inner) + \sum_{\numconstr = 1}^\numconstrs \frac{- \grad[{\inner[\numconstr]}] \obj(\outer, \inner)}{\grad[{\inner[\numconstr]}] \constr[\numconstr](\outer, \inner[\numconstr])} \constr[\numconstr](\outer, \inner[\numconstr])$ is convex as well.


\end{proof}

\propstackelbergconvex*
\begin{proof}[Proof of \Cref{prop:stackelberg_convex}]
By \Cref{stackleberg-existence}, the set of Stackelberg equilibria of any min-max Stackelberg game is non-empty. Additionally, under \Cref{main-assum}, we have that for all $\outer \in \outers$, $\obj(\outer, \cdot)$ is concave, and $\{\inner \in \innerset \mid \constr(\outer, \inner) \geq \zeros\}$ is convex. Hence, by Theorem 2.6 of Rockafeller \cite{rockafellar2009variational}, the set of solutions $\innerset^*(\outer) = \argmax_{\inner \in \innerset: \constr(\outer, \inner) \geq \zeros} \obj(\outer, \inner)$ is compact- and convex- valued.
Similarly, by \Cref{thm:convex-value-func}, under \Cref{assum:existence}, $\max_{\inner \in \innerset: \constr(\outer, \inner) \geq \zeros} \obj(\outer, \inner)$ is continuous and convex. Hence, the set of solutions $\outerset^* = \argmin_{\outer \in \outerset} \max_{\inner \in \innerset: \constr(\outer, \inner) \geq \zeros} \obj(\outer, \inner)$ is compact- and convex-valued. Since the composition of two compact-convex-valued correspondences is again compact-convex-valued (Proposition 5.52 of Rockafeller \cite{rockafellar2009variational}), we conclude that the set of Stackelberg equilibria, namely $\outerset^*(\innerset^*)$, is compact and convex.
\end{proof}

\section{Envelope Theorem}
\label{sec-app:envelope}

Danskin's theorem \cite{danskin1966thm} offers insights into optimization problems of the form:
    $\max_{\inner \in \innerset} \obj (\outer, \inner)$,
where $\innerset \subset \R^\innerdim$ is compact and non-empty.
Among other things, Danskin's theorem allows us to compute the gradient of the objective function of this optimization problem with respect to $\outer$.
\begin{theorem}[Danskin's Theorem]
\label{Danskinsthm}
    Consider an optimization problem of the form:
    $\max_{\inner \in \innerset} \obj (\outer, \inner)$, where $\innerset \subset \R^\innerdim$ is compact and non-empty.
    Suppose that $\innerset$ is convex and that $\obj$ is concave in $\inner$. Let $\val(\outer) = \max_{\inner \in \innerset} \obj (\outer, \inner)$ and 
    $\innerset^*(\outer) = \argmax_{\inner \in \innerset} \obj (\outer, \inner)$. Then $\val$ is differentiable at $\widehat{\outer}$, if the solution correspondence $\innerset^*(\widehat{\outer})$ is a singleton: i.e., $\innerset^*(\widehat{\outer}) = \{ \inner^*(\widehat{\outer}) \}$. Additionally, the gradient at $\widehat{\outer}$ is given by
        $\val^{\prime}(\widehat{\outer}) = \grad[\outer] \obj (\widehat{\outer}, \inner^*(\widehat{\outer}))$.
\end{theorem}

Unfortunately, Danskin's theorem does not hold when the set $\innerset$ is replaced by a correspondence, 
which occurs in min-max Stackelberg games: i.e., when the inner problem is $\max_{\inner \in \innerset: \constr (\outer, \inner) \geq \zeros} \obj (\outer, \inner)$.

\begin{example}[Danskin's theorem does not apply to min-max Stackelberg games]
\label{no-danskins-ex}
Consider the optimization problem:
\begin{align}
    \max_{\inner[ ] \in \R :  \inner[ ] + \outer[ ] \geq 0} -\inner[ ]^2 + \inner[ ] + 2\outer[ ] + 2 \enspace .
\end{align}
 
The solution to this problem is unique, given any $\outer[ ] \in \outerset$, meaning the solution correspondence $\innerset^* (\outer[ ])$ is singleton-valued.
We denote this unique solution by $\inner[ ]^* (\outer[ ])$.
After solving, we find that
\begin{align}
    \inner[ ]^*(\outer[ ]) = \left\{\begin{array}{cc}
        \nicefrac{1}{2} & \text{ if } \outer[ ] \geq - \nicefrac{1}{2} \\
        -\outer[ ] & \text{ if } \outer[ ] < - \nicefrac{1}{2}
    \end{array}\right.
\end{align}
The value function $V(\outer[ ]) = \max_{\inner[ ] \in \R :  \inner[ ] + \outer[ ] \geq 0} -\inner[ ]^2 + \inner[ ] + 2\outer[ ] + 2$ is then given by:
\begin{align}
    V(\outer[ ]) &= f(\outer[ ], \inner[ ]^*(\outer[ ]))\\
    &=-\inner[ ]^*(\outer[ ])^2 + \inner[ ]^*(\outer[ ]) + 2\outer[ ] + 2\\
    &= \left\{\begin{array}{cc}
        -\nicefrac{1}{4} + \nicefrac{1}{2} + 2 \outer[ ] +2 & \text{ if } \outer[ ] \geq - \nicefrac{1}{2} \\
        -\outer[ ]^2 - \outer[ ] + 2\outer[ ] + 2 & \text{ if  }  \outer[ ] < - \nicefrac{1}{2}
    \end{array}\right.\\
    &= \left \{\begin{array}{cc}
         \nicefrac{9}{4} + 2 \outer[ ] & \text{ if } \outer[ ] \geq - \nicefrac{1}{2} \\
        -\outer[ ]^2 + \outer[ ]  + 2 & \text{ if  }  \outer[ ] < - \nicefrac{1}{2}
    \end{array} \right.
\label{value-func-example}
\end{align}
The derivative of this value function is:
\begin{align}
    \frac{\partial V}{\partial \outer[ ]} &= \left \{\begin{array}{cc}
         2 & \text{ if } \outer[ ] \geq - \nicefrac{1}{2} \\
        1 - 2 \outer[ ]  & \text{ if  }  \outer[ ] < - \nicefrac{1}{2}
    \end{array} \right.
\label{value-func-deriv-example}
\end{align}
However, the derivative predicted by Danskin's theorem is 2.
Hence, Danskin's theorem does not hold when the constraints are parameterized, i.e., when the problem is of the form $\min_{\inner \in \innerset: \constr (\outer, \inner)} \obj (\outer, \inner)$ rather than $\min_{\inner \in \innerset} \obj (\outer, \inner)$, where $\outerset \subset \R^\outerdim$, $\innerset \subset \R^\innerdim$,  and for all $\numconstr \in [\numconstrs]$, $\constr[\numconstr]: \outerset \times \innerset \to \R$ are continuous.


\textbf{N.B.} For simplicity, we do not assume the constraint set is compact in this example.
Compactness of the constraint set can be used to guarantee existence of a solution,
but as a solution to this particular problem always exists, we can do away with this assumption.
\end{example}


The following theorem, due to \citet{milgrom2002envelope}, generalizes Danskin's theorem to handle parameterized constraints:

\begin{theorem}[Envelope Theorem \cite{milgrom2002envelope}]\label{envelope}
    Consider the maximization problem
    \begin{align}\label{envelope-max}
        \val(\outer) = \max _{\inner \in \innerset} \obj (\outer, \inner), \text { subject to } \constr[k](\outer, \inner) \geq 0, \text { for all } \numconstr =1, \ldots, \numconstrs \enspace ,
    \end{align}
    where $\innerset \subseteq \R^{\innerdim}$.
    
    Define the solution correspondence 
    $\innerset^*(\outer) = \argmax _{\inner \in \outerset:  \constr (\outer, \inner) \geq \zeros} \obj (\outer, \inner)$.
    If \Cref{main-assum} holds, then the value function $\val$ is absolutely continuous, and at any point $\widehat{\outer} \in \outerset$ where $\val$ is differentiable:
    \begin{align}
    \label{envelope-deriv}
        \grad[\outer] \val(\widehat{\outer}) = \grad[\outer] L(\inner^*(\widehat{\outer}), \langmult^*(\widehat{\outer}, \inner^*(\widehat{\outer}))), \widehat{\outer}) = \grad[\outer] \obj\left( \widehat{\outer}, \inner^{*}(\widehat{\outer})\right) +\sum_{\numconstr = 1}^{\numconstrs} \langmult[\numconstr]^*(\widehat{\outer}, \inner^*(\widehat{\outer})) \grad[\outer] \constr[k]\left(\widehat{\outer}, \inner^{*}(\widehat{\outer})\right) \enspace ,
    \end{align}
    where $\langmult^*(\widehat{\outer}, \inner^*(\widehat{\outer})) = \left(\langmult[1](\widehat{\outer}, \inner^*(\widehat{\outer})), \ldots, \langmult[\numconstrs]^*(\widehat{\outer}, \inner^*(\widehat{\outer})) \right)^T \in \langmults^*(\widehat{\outer}, \inner^*(\widehat{\outer}))$ are the KKT multipliers associated associated with $\inner^{*}(\widehat{\outer}) \in \innerset^*(\widehat{\outer})$. 
\end{theorem}

\section{Omitted Subdifferential Envelope Theorem Proof (\Cref{sec:envelope})}

\begin{proof}[Proof of \Cref{envelope-sd}]
As usual, let $\val(\outer) = \max_{\inner \in \innerset : \constr (\outer, \inner) \geq 0} \obj (\outer, \inner)$. First, note that \Cref{thm:convex-value-func} $\val$ is subdifferentiable as it is convex \cite{subgradients-notes}.
Reformulating the problem as a Lagrangian saddle point problem, for all $\widehat{\outer} \in \outerset$, it holds that:
\begin{align}
    \val(\widehat{\outer}) &= \max_{\inner \in \innerset: \constr (\widehat{\outer}, \inner) \geq 0} \obj (\widehat{\outer}, \inner)\\ 
    &= \max_{\inner \in \innerset} \min_{\langmult \in \R^\numconstrs_{++}} \left\{ \obj (\widehat{\outer}, \inner) + \sum_{\numconstr = 1}^\numconstrs \langmult[\numconstr] \constr[\numconstr](\widehat{\outer}, \inner) \right\}\label{eq:exists-langmult}
\end{align}
Since $\obj$ is continuous, $\innerset$ is compact, and $\constr[1], \hdots, \constr[\numconstrs]$ are continuous, for all $\widehat{\outer} \in \outerset$, there exists $\inner^{*}(\widehat{\outer}) \in \argmax_{\inner \in \innerset: \constr(\widehat{\outer}, \inner) \geq \zeros} \obj(\widehat{\outer}, \inner)$.
%
Furthermore, as \Cref{main-assum} ensures that an interior solution exists, the Karush-Kuhn-Tucker Theorem \cite{kuhn1951kkt} applies, so for all $\widehat{\outer} \in \outerset$ and any associated $\inner^{*}(\widehat{\outer})$, there exists $\langmult(\widehat{\outer}, \inner^{*}(\widehat{\outer})) \in \R^\numconstrs$ 
that solves \Cref{eq:exists-langmult}.

Define the solution correspondence $\innerset^*(\outer) = \argmax_{\inner \in \innerset: \constr (\outer, \inner) \geq 0} \obj (\outer, \inner)$, and let $\langmults^*(\outer, \inner) = \argmin_{\langmult \in \R^\numconstrs_{+}} \left\{ \obj (\outer, \inner) + \sum_{\numconstr = 1}^\numconstrs \langmult[\numconstr] \constr[\numconstr](\outer, \inner) \right\}$.
We can then re-express the value function at $\widehat{\outer}$ as:
\begin{align*}
    \val(\widehat{\outer}) = \obj (\widehat{\outer}, \inner^*(\widehat{\outer})) + \sum_{\numconstr = 1}^\numconstrs \langmult[\numconstr]^*(\widehat{\outer}, \inner^*(\widehat{\outer})) \constr[\numconstr](\widehat{\outer}, \inner^*(\widehat{\outer})), && \forall \inner^*(\widehat{\outer}) \in \innerset^*(\widehat{\outer}), \langmult[\numconstr]^*(\widehat{\outer},  \inner^*(\widehat{\outer})) \in \langmults[\numconstr]^*(\widehat{\outer}, \inner^*(\widehat{\outer})) \enspace .
\end{align*}
Equivalently, we can take the maximum over $\inner^*$'s and $\langmult^*$'s to obtain:
\begin{align*}
    \val(\widehat{\outer})&= \max_{\inner^*(\widehat{\outer}) \in \innerset^*(\widehat{\outer})}  \max_{\langmult[\numconstr]^*(\widehat{\outer}, \inner^*(\widehat{\outer})) \in \langmults[\numconstr]^*(\widehat{\outer}, \inner^*(\widehat{\outer}))} \left\{ \obj (\widehat{\outer}, \inner) + \sum_{\numconstr = 1}^\numconstrs \langmult[\numconstr]^*(\widehat{\outer}, \inner^*(\widehat{\outer})) \constr[\numconstr](\widehat{\outer}, \inner) \right\} \enspace .
\end{align*}

Note that for fixed $\inner^*(\widehat{\outer}) \in \innerset^*(\widehat{\outer})$ and corresponding fixed $\langmult[\numconstr]^*(\widehat{\outer}, \inner^*(\widehat{\outer})) \in \langmults[\numconstr]^*(\widehat{\outer}, \inner^*(\widehat{\outer}))$, \\ $\obj (\widehat{\outer}, \inner) + \sum_{\numconstr = 1}^\numconstrs \langmult[\numconstr]^*(\widehat{\outer}, \inner^*(\widehat{\outer})) \constr[\numconstr](\widehat{\outer}, \inner^*(\widehat{\outer}))$ is differentiable, since $\funcs$ are differentiable.

Additionally, recall the pointwise maximum subdifferential property,%
\footnote{See, for example, \cite{subgradients-notes}.}
i.e., if $f(\outer) = \max_{\alpha \in \calA} f_\alpha(\outer)$ for a family of functions $\{f_\alpha\}_{\alpha \in \calA}$, then $\subdiff[\outer]f(\a) = \mathrm{conv} \left(\bigcup_{\alpha \in \calA} \left\{\subdiff[\outer]f_{\alpha \in \calA}(\a) \mid f_\alpha(\a) = f(\outer)  \right\}\right)$, which then gives:
\begin{align}
    \subdiff[\outer] \val(\widehat{\outer}) &= \subdiff[\outer] \left( \max_{\inner^*(\widehat{\outer}) \in \innerset^*(\widehat{\outer})}  \max_{\langmult[\numconstr]^*(\widehat{\outer}, \inner^*(\widehat{\outer})) \in \langmults[\numconstr]^*(\widehat{\outer}, \inner^*(\widehat{\outer}))} \left\{ \obj (\widehat{\outer}, \inner^*(\widehat{\outer}))  + \sum_{\numconstr = 1}^\numconstrs \langmult[\numconstr]^*(\widehat{\outer}, \inner^*(\widehat{\outer})) \constr[\numconstr](\widehat{\outer}, \inner^*(\widehat{\outer})) \right\} \right)\\
    &= \mathrm{conv} \left( \bigcup_{\inner^*(\widehat{\outer}) \in \innerset^*(\widehat{\outer})}      \bigcup_{\langmult[\numconstr]^*(\widehat{\outer}, \inner^*(\widehat{\outer})) \in \langmults[\numconstr]^*(\widehat{\outer}, \inner^*(\widehat{\outer}))} \subdiff[\outer] \left\{ \obj (\widehat{\outer}, \inner^*(\widehat{\outer}))  + \sum_{\numconstr = 1}^\numconstrs \langmult[\numconstr]^*(\widehat{\outer}, \inner^*(\widehat{\outer})) \constr[\numconstr](\widehat{\outer}, \inner^*(\widehat{\outer})) \right\}
    \right)\\
    &= \mathrm{conv} \left( \bigcup_{\inner^*(\widehat{\outer}) \in \innerset^*(\widehat{\outer})}      \bigcup_{\langmult[\numconstr]^*(\widehat{\outer}, \inner^*(\widehat{\outer})) \in \langmults[\numconstr]^*(\widehat{\outer}, \inner^*(\widehat{\outer}))} \left\{ \grad[\outer] \obj\left( \widehat{\outer}, \inner^{*}(\widehat{\outer})\right) + \sum_{\numconstr = 1}^{\numconstrs} \langmult[\numconstr]^*(\widehat{\outer}, \inner^*(\widehat{\outer})) \grad[\outer] \constr[k]\left(\widehat{\outer}, \inner^{*}(\widehat{\outer})\right) \right\}
    \right) \enspace .
\end{align}
\end{proof}

\section{Convergence Results for \Cref{sec:envelope}}
\label{sec-app:algos}

\begin{proof}[Proof of \Cref{l-mogd}]
By our subdifferential envelope theorem (\Cref{envelope-sd}), we have:
\begin{align}
\grad[\outer] \obj (\outer^{(\outeriter-1)}, \inner^{(\outeriter-1)}) + \sum_{\numconstr = 1}^\numconstrs \langmult[\numconstr]^{(\outeriter-1)} \grad[\outer]  \constr[\numconstr](\outer^{(\outeriter-1)}, \inner^{(\outeriter-1)}) 
&\in  \subdiff[\outer] \val(\outer^{(t-1)}) \\
&= \subdiff[\outer] \max_{\inner \in \innerset: \constr (\outer^{(\outeriter-1)}, \inner) \geq 0} \obj (\outer^{(\outeriter-1)}, \inner) \enspace .
\end{align}

For notational clarity, let $\subgrad (\outeriter-1) = \grad[\outer] \obj (\outer^{(\outeriter-1)}, \inner^{(\outeriter-1)}) + \sum_{\numconstr = 1}^\numconstrs \langmult[\numconstr]^{(\outeriter-1)} \grad[\outer]  \constr[\numconstr](\outer^{(\outeriter-1)}, \inner^{(\outeriter-1)})$.
Suppose that $\outer^{*} \in \argmin_{\outer \in \outerset} \max_{\inner \in \innerset: \constr (\outer, \inner) \geq \zeros} \obj (\outer, \inner)$.
Then:
\begin{align}
    &\left\| \outer^{(\outeriters)} - \outer^{*} \right\|^2 \\ 
    &= \left\| \project[\outerset] \left( \outer^{(\outeriters-1)} - \learnrate[\outeriters] \subgrad (\outeriters-1) \right) - \project[\outerset] \left( \outer^{*} \right) \right\|^2 \\
    &\leq \left\| \outer^{(\outeriters-1)} - \learnrate[\outeriters] \subgrad (\outeriters-1) -  \outer^{*} \right\|^2 \\
    &= \left\| \outer^{(\outeriters-1)} - \outer^{*} \right\|^2 - 2\learnrate[\outeriters] \left<\subgrad (\outeriters-1),\left( \outer^{(\outeriters-1)} - \outer^{*} \right) \right> + \learnrate[\outeriters][2] \left\| \subgrad (\outeriters-1) \right\|^2 \enspace ,
\end{align}

\noindent
where the first line follows from the subgradient descent rule and the fact that $\outer^{*} \in \outerset$;
the second, because the project operator is a non-expansion; and
the third, by the definition of the norm.

Let for any $\outeriter \in \N_+$, $\err[\outeriter] = \left<\subgrad (\outeriter) - \grad \val(\outer^{(\outeriter)}),\left( \outer^{(\outeriter)} - \outer^{*} \right) \right>$:
\begin{align}
    &= \left\| \outer^{(\outeriter-1)} - \outer^{*} \right\|^2 - 2\learnrate[\outeriter] \left<\grad \val(\outer^{(\outeriter- 1)}),\left( \outer^{(\outeriter-1)} - \outer^{*} \right) \right> - 2\learnrate[\outeriter] \left<\subgrad (\outeriter-1) - \grad \val(\outer^{(\outeriter- 1)}),\left( \outer^{(\outeriter-1)} - \outer^{*} \right) \right> \notag \\ 
    &+ \learnrate[\outeriter][2] \left\| \subgrad (\outeriter-1) \right\|^2 \\
    &= \left\| \outer^{(\outeriter-1)} - \outer^{*} \right\|^2 - 2\learnrate[\outeriter] \left<\grad \val(\outer^{(\outeriter- 1)}),\left( \outer^{(\outeriter-1)} - \outer^{*} \right) \right> - 2\learnrate[\outeriter] \err[\outeriter - 1] + \learnrate[\outeriter][2] \left\| \subgrad (\outeriter-1) \right\|^2 \\
    &\leq \left\| \outer^{(\outeriter-1)} - \outer^{*} \right\|^2 - 2\learnrate[\outeriter] \left( \val (\outer^{(\outeriter-1)}) - \val (\outer^{*}) \right) - 2\learnrate[\outeriter] \err[\outeriter - 1] + \learnrate[\outeriter][2] \left\| \subgrad (\outeriter-1) \right\|^2 \enspace ,
\end{align}
\noindent where the last line follows by the definition of the subgradient, i.e., $\left< \grad \val(\outer^{(\outeriter- 1)}), \left( \outer^{(\outeriter-1)} - \outer^{*} \right) \right> \geq \obj (\outer^{(\outeriter-1)}, \inner^{(\outeriter-1)}) - \obj (\outer^{*}, \inner^{(\outeriter-1)})$. 
Applying this inequality recursively, we obtain:
\begin{align}
    \left\| \outer^{(\outeriter)} - \outer^{*} \right\|^2 \leq \left\| \outer^{(0)} - \outer^{*} \right\|^2 - \sum_{\outeriter = 1}^{\outeriter} 2\learnrate[\outeriter] \left( \val (\outer^{(\outeriter-1)}) - \val (\outer^{*}) \right) - \sum_{\outeriter = 1}^{\outeriter} 2\learnrate[\outeriter] \err[\outeriter - 1] + \sum_{\outeriter = 1}^{\outeriter} \learnrate[\outeriter][2] \left\| \subgrad (\outeriter-1) \right\|^2 \enspace .\label{eq:leftover_rec}
\end{align}
Since $\left\| \outer^{(\outeriter)} - \outer^{*} \right\| \geq 0$, re-organizing, we have:
\begin{align}
    2 \sum_{\outeriter = 1}^{\outeriters} \learnrate[\outeriter] \left( \val (\outer^{(\outeriters-1)}) - \val (\outer^{*}) \right)  \leq  \left\| \outer^{(0)} - \outer^{*} \right\|^2 - \sum_{\outeriter = 1}^{\outeriters} 2\learnrate[\outeriters] \err[\outeriters - 1]  + \sum_{\outeriter = 1}^{\outeriters} \learnrate[\outeriter][2] \left\| \subgrad (\outeriter-1) \right\|^2 \enspace .
\end{align}

Let $(\bestiter[x][t], \bestiter[y][t]) = (\outer^{(k^*)}, \inner^{(k^*)})$ where $k^* \in \argmin_{ k \in [t]} \val (\outer^{(k)})$. Then:
\begin{align}
    & \sum_{\outeriter = 1}^{\outeriters} \learnrate[\outeriter] \left( \val (\outer^{(\outeriters-1)}) - \val (\outer^{*}) \right) \\
    &\geq \left( \sum_{\outeriter = 1}^{\outeriters} \learnrate[\outeriter] \right) \min_{\outeriter \in [\outeriters]} \left( \val (\outer^{(\outeriters-1)}) - \val (\outer^{*}) \right) \\
    &= \left( \sum_{\outeriter = 1}^{\outeriters} \learnrate[\outeriter] \right) \left( \val (\bestiter[x][\outeriters - 1]) - \val (\outer^{*}) \right)
\end{align}

\noindent
Combining the above inequality with \Cref{eq:leftover_rec}, we get the following bound:
\begin{align}
    \val (\bestiter[x][\outeriters - 1]) - \val (\outer^{*}) \leq \frac{\left\| \outer^{(0)} - \outer^{*} \right\|^2 - \sum_{\outeriter = 1}^{\outeriters} 2\learnrate[\outeriter] \err[\outeriter - 1]  + \sum_{\outeriter = 1}^{\outeriters} \learnrate[\outeriter][2]||\subgrad (\outeriter-1)||^2}{2 \left( \sum_{\outeriter = 1}^{\outeriters} \learnrate[\outeriter] \right)}
\end{align}

Now,
since the value function $\val$ is $\lipschitz[\val]$-Lipschitz continuous, where $\lipschitz[\val] = \max_{(\widehat{\outer}, \widehat{\inner}) \in \outerset \times \innerset} \left\| \grad[\outer] \obj\left( \widehat{\outer}, \inner^{*}(\widehat{\outer})\right) + \sum_{\numconstr = 1}^{\numconstrs} \langmult[\numconstr]^*(\widehat{\outer}, \inner^*(\widehat{\outer})) \grad[\outer] \constr[k]\left(\widehat{\outer}, \inner^{*}(\widehat{\outer})\right)\right\|$ and $\langmult^*(\widehat{\outer}, \inner^*(\widehat{\outer}))$ are the optimal KKT multipliers associated with $\inner^{*}(\widehat{\outer}) \in \argmax_{\inner \in \innerset:  \constr(\widehat{\outer}, \inner) \geq \zeros} \obj(\outer, \inner)$,  all the subgradients are bounded: i.e., for all $k \in \N, \left\| \subgrad (k-1) \right\| \leq \lipschitz[\val]$.
So:
\begin{align}
    \val (\bestiter[x][T -1 ]) - \val (\outer^{*}) &\leq \frac{\left\| \outer^{(0)} - \outer^{*} \right\|^2 - \sum_{\outeriter = 1}^{\outeriters} 2\learnrate[\outeriter] \err[\outeriter - 1]  + \lipschitz[\val]^2 \sum_{\outeriter = 1}^{\outeriters} \learnrate[\outeriter][2]}{2 \left( \sum_{\outeriter = 1}^{\outeriters} \learnrate[\outeriter] \right)}\\
    & \frac{\left\| \outer^{(0)} - \outer^{*} \right\|^2 + \sum_{\outeriter = 1}^{\outeriters} 2\learnrate[\outeriter] \left|\err[\outeriter - 1] \right|  + \lipschitz[\val]^2 \sum_{\outeriter = 1}^{\outeriters} \learnrate[\outeriter][2]}{2 \left( \sum_{\outeriter = 1}^{\outeriters} \learnrate[\outeriter] \right)}
\end{align}

Letting $\overline{\err[\outeriters]} \doteq \frac{\sum_{\outeriter = 1}^{\outeriters} \learnrate[\outeriter] \left|\err[\outeriter - 1] \right|}{\sum_{\outeriter = 1}^{\outeriters} \learnrate[\outeriter]}$, we get: 
\begin{align}
    \obj (\bestiter[x][T], \bestiter[x][T]) - \min_{\outer \in \outerset} \max_{\inner \in \innerset : \constr (\outer, \inner) \geq 0}  \obj (\outer, \inner) &\leq \frac{\left\| \outer^{(0)} - \outer^{*} \right\|^2 + \lipschitz[\val] \sum_{\outeriter = 1}^{\outeriters} \learnrate[\outeriter][2]}{2 \left( \sum_{\outeriter = 1}^{\outeriters} \learnrate[\outeriter] \right)} + \overline{\err[\outeriters]}\label{eq:conv-upper-bound}
\end{align}

Recall the assumptions that the step sizes are square-summable but not summable, namely
    $\sum_{k = 1}^{\outeriters} \learnrate[k][2] \leq \infty$
    and
    $\sum_{k = 1}^{\outeriters} \learnrate[k] = \infty$.
%
Now as $\outeriters \to \infty$, \Cref{eq:conv-upper-bound} becomes:
\begin{align}
\lim_{k \to \infty} \obj (\bestiter[x][k], \bestiter[y][k]) \leq \min_{\outer \in \outerset} \max_{\inner \in \innerset : \constr (\outer, \inner) \geq 0}  \obj (\outer, \inner) + \overline{\err[\outeriters]}
\enspace .
\end{align}
%
We have thus proven the first inequality of the two that define an $(\overline{\err[\outeriters]},\delta)$-Stackelberg equilibrium.

The second inequality follows by construction, as for all $k \in \N$, the max oracle returns $\bestiter[y][k]$ that satisfies $\obj (\bestiter[x][k], \bestiter[y][k]) \geq  \max_{\inner \in \innerset : \constr (\outer, \inner) \geq 0}  \obj (\bestiter[x][k], \inner) - \delta$.
Thus, as $k \to \infty$, the best iterate converges to an $(\overline{\err[\outeriters]}, \delta)$-Stackelberg equilibrium.
%
Additionally, setting $\learnrate[\outeriter] = \frac{\left\| \outer^{(0)} - \outer^{*} \right\|}{ \sqrt{\outeriters}}$, we see that for all $\outeriter \in [T]$,
\begin{align}
    \obj (\bestiter[x][T], \bestiter[x][T]) - \min_{\outer \in \outerset} \max_{\inner \in \innerset : \constr (\outer, \inner) \geq 0} \obj (\outer, \inner)  &\leq  \frac{ \left\| \outer^{(0)} - \outer^{*} \right\|^2}{\sqrt{\outeriters}} + \overline{\err[\outeriters]} \enspace .
\end{align}
Likewise, setting $\varepsilon \leq \frac{\left\| \outer^{(0)} - \outer^{*} \right\|^2}{\sqrt{\outeriters}}$,
we obtain $\outeriters \leq \frac{\left\| \outer^{(0)} - \outer^{*} \right\|^4}{\varepsilon^2}$,
implying that
the best iterate converges to an $(\varepsilon + \overline{\err[\outeriters]}, \delta)$-Stackelberg equilibrium in $O(\varepsilon^{-2})$ iterations.

Finally, by the Cauchy-Schwarz inequality, $|\err[\outeriters - 1]| \leq \left\|\subgrad (\outeriter) - \grad \val(\outer^{(\outeriter)})\right\| \left\|\outer^{(\outeriter)} - \outer^{*} \right\| 
$, giving us the theorem statement.
\end{proof}

\begin{theorem}
\label{ls-mogd}
Suppose \Cref{mogd} is run on min-max Stackelberg game given by $(\outerset, \innerset, \obj, \constr)$ which satisfies \Cref{main-assum}, and that $\val$ is $\mu$-strongly convex in $\outer$. Then, if $(\bestiter[x][t], \bestiter[y][t]) \in \argmin_{(\outer^{(k)}, \inner^{(k)}) : k \in [t]} \obj (\outer^{(k)}, \inner^{(k)})$, for $\varepsilon \in (0,1)$, and for all $\outeriter \in \outeriters$ $\learnrate[\outeriter] = \frac{2}{\mu(\outeriter +1)}$, if we choose
%
    $\outeriters \geq
    N_\outeriters (\varepsilon) \in O(\varepsilon^{-1})$,
%
then there exists an iteration $\outeriters^{*} \leq \outeriters$ s.t.\ $(\bestiter[x][\outeriters^{*}], \bestiter[y][\outeriters^{*}])$ is an $(\frac{8\lipschitz[\val](c)}{\mu}, \delta)$-Stackelberg equilibrium.
\end{theorem}

\begin{proof}[Proof of \Cref{ls-mogd}]

For notational clarity, let $\subgrad (\outeriter-1) = \grad[\outer] \obj (\outer^{(\outeriter-1)}, \inner^{(\outeriter-1)}) + \sum_{\numconstr = 1}^\numconstrs \langmult[\numconstr]^{(\outeriter-1)} \grad[\outer]  \constr[\numconstr] (\outer^{(\outeriter-1)}, \inner^{(\outeriter-1)})$.
Suppose that $\outer^{*} \in \argmin_{\outer \in \outerset} \max_{\inner \in \innerset: \constr (\outer, \inner)} \obj (\outer, \inner)$.
Then, for all $\outeriter \in \N$ s.t.\ $\outeriter \geq 1$, we have:
\begin{align}
    & \left\| \outer^{(\outeriter)} - \outer^{*} \right\|^2 \\ 
    &= \left\| \project[\outerset] \left( \outer^{(\outeriter-1)} - \learnrate[\outeriter] \subgrad (\outeriter-1) \right) - \project[\outerset] \left( \outer^{*} \right) \right\|^2 \\
    &\leq \left\| \outer^{(\outeriter-1)} - \learnrate[\outeriter] \subgrad (\outeriter-1) -  \outer^{*} \right\|^2 \\
    &= \left\| \outer^{(\outeriter-1)} - \outer^{*} \right\|^2 - 2\learnrate[\outeriter] \left<\subgrad (\outeriter-1),\left( \outer^{(\outeriter-1)} - \outer^{*} \right) \right> + \learnrate[\outeriter][2] \left\| \subgrad (\outeriter-1) \right\|^2
\end{align}
where the first line follows from the subgradient method's update rule and the fact that $\outer^{*} \in \outerset$;
the second, because the project operator is a non-expansion; and
the third, by the definition of the norm.

Let for any $\outeriter \in \N_+$, $\err[\outeriter] = \left<\subgrad (\outeriter) - \grad \val(\outer^{(\outeriter)}),\left( \outer^{(\outeriter)} - \outer^{*} \right) \right>$:
\begin{align}
    &= \left\| \outer^{(\outeriter-1)} - \outer^{*} \right\|^2 - 2\learnrate[\outeriter] \left<\grad \val(\outer^{(\outeriter- 1)}),\left( \outer^{(\outeriter-1)} - \outer^{*} \right) \right> - 2\learnrate[\outeriter] \left<\subgrad (\outeriter-1) - \grad \val(\outer^{(\outeriter- 1)}),\left( \outer^{(\outeriter-1)} - \outer^{*} \right) \right> \notag \\ 
    &+ \learnrate[\outeriter][2] \left\| \subgrad (\outeriter-1) \right\|^2 \\
    &= \left\| \outer^{(\outeriter-1)} - \outer^{*} \right\|^2 - 2\learnrate[\outeriter] \left<\grad \val(\outer^{(\outeriters- 1})),\left( \outer^{(\outeriter-1)} - \outer^{*} \right) \right> - 2\learnrate[\outeriter] \err[\outeriter - 1] + \learnrate[\outeriter][2] \left\| \subgrad (\outeriter-1) \right\|^2 \\
    &\leq \left\| \outer^{(\outeriter-1)} - \outer^{*} \right\|^2 - 2\learnrate[\outeriter] \left[ \frac{\mu}{2} \left\| \outer^{(\outeriter-1)} - \outer^{*} \right\|^2 + \val(\outer^{(\outeriter-1)}) - \val(\outer^{*})  \right] -2\learnrate[\outeriter] \err[\outeriter - 1] + \learnrate[\outeriter][2] \left\| \subgrad (\outeriter-1) \right\|^2 \\
    &= \left\| \outer^{(\outeriter-1)} - \outer^{*} \right\|^2 - \learnrate[\outeriter] \mu \left\| \outer^{(\outeriter-1)} - \outer^{*} \right\|^2 - 2\learnrate[\outeriter]  \left( \val (\outer^{(\outeriter-1)}) - \val (\outer^{*}) \right) -2\learnrate[\outeriter] \err[\outeriter - 1] + \learnrate[\outeriter][2] \left\| \subgrad (\outeriter-1) \right\|^2 \\
    &= \left(1 - \learnrate[\outeriter] \mu \right) \left\| \outer^{(\outeriter-1)} - \outer^{*} \right\|^2  - 2\learnrate[\outeriter]  \left( \val (\outer^{(\outeriter-1)}) - \val (\outer^{*}) \right) - 2\learnrate[\outeriter] \err[\outeriter - 1] + \learnrate[\outeriter][2] \left\| \subgrad (\outeriter-1) \right\|^2\label{eq:organized-bound-d2}
\end{align}

Re-organizing \Cref{eq:organized-bound-d2}, yields:
\begin{align}
    &\val (\outer^{(\outeriter-1)}) - \val (\outer^{*}) \\
    &\leq \frac{1 - \learnrate[\outeriter] \mu}{2\learnrate[\outeriter]} \left\| \outer^{(\outeriter-1)} - \outer^{*} \right\|^2 - \frac{1}{2\learnrate[\outeriter]} \left\| \outer^{(\outeriter)} - \outer^{*} \right\|^2 - 2 \learnrate[\outeriter]\err[\outeriter - 1] + \frac{\learnrate[\outeriter]}{2} \left\| \subgrad (\outeriter-1) \right\|^2 \enspace .
\end{align}
Next, setting for all $\outeriter \in \N_+$ $\learnrate[\outeriter] = \frac{2}{\mu(\outeriter+1)}$, we get:
\begin{align}
    &\val (\outer^{(\outeriter-1)}) - \val (\outer^{*}) 
    \\
    &\leq \frac{\mu(\outeriter-1)}{4} \left\| \outer^{(\outeriter-1)} - \outer^{*} \right\|^2 - \frac{\mu(\outeriter+1)}{4} \left\| \outer^{(\outeriter)} - \outer^{*} \right\|^2 - \frac{4}{\mu(\outeriter+1)} \err[\outeriter - 1] + \frac{1}{\mu(\outeriter+1)} \left\| \subgrad (\outeriter-1) \right\|^2 \enspace .
\end{align}    
Multiplying both sides by $\outeriter$, we now have:
\begin{align}    
    &\outeriter \left( \val (\outer^{(\outeriter-1)}) - \val (\outer^{*}) \right) \\
    &\leq \frac{\mu\outeriter(\outeriter-1)}{4} \left\| \outer^{(\outeriter-1)} - \outer^{*} \right\|^2 - \frac{\mu\outeriter(\outeriter+1)}{4} \left\| \outer^{(\outeriter)} - \outer^{*} \right\|^2 - \frac{4\outeriter}{\mu(\outeriter+1)} \err[\outeriter - 1] + \frac{1}{\mu} \left\| \subgrad (\outeriter-1) \right\|^2 \enspace . 
\end{align}

Summing up across all iterations on both sides:
\begin{align}
    & \sum_{\outeriter = 1}^{\outeriters } \outeriter \left( \val (\outer^{(\outeriter - 1)}) - \val (\outer^{*}) \right) \\
    &\leq \sum_{\outeriter = 1}^{\outeriters } \frac{\mu\outeriter(\outeriter-1)}{4} \left\| \outer^{(\outeriter-1)} - \outer^{*} \right\|^2 - \sum_{\outeriter = 1}^{\outeriters } \frac{\mu\outeriter(\outeriter+1)}{4} \left\| \outer^{(\outeriter)} - \outer^{*} \right\|^2 - \sum_{\outeriter = 1}^{\outeriters }\frac{4\outeriter}{\mu(\outeriter+1)} \err[\outeriter - 1] + \sum_{\outeriter = 1}^{\outeriters}  \frac{1}{\mu} \left\| \subgrad (\outeriter-1) \right\|^2 \\
    &= \sum_{\outeriter = 1}^{\outeriters } \frac{\mu\outeriter(\outeriter-1)}{4} \left\| \outer^{(\outeriter-1)} - \outer^{*} \right\|^2 - \sum_{\outeriter = 2}^{\outeriters + 1} \frac{\mu(\outeriter-1) \outeriter}{4} \left\| \outer^{(\outeriter-1)} - \outer^{*} \right\|^2 - \sum_{\outeriter = 1}^{\outeriters }\frac{4\outeriter}{\mu(\outeriter+1)}\err[\outeriter - 1] + \sum_{\outeriter = 1}^{\outeriters}  \frac{1}{\mu} \left\| \subgrad (\outeriter-1) \right\|^2 \\
    &= - \frac{\mu\outeriters(\outeriters+1)}{4} \left\| \outer^{(\outeriters )} - \outer^{*} \right\|^2 -\sum_{\outeriter = 1}^{\outeriters }\frac{4\outeriter}{\mu(\outeriter+1)}\err[\outeriter - 1] + \sum_{\outeriter = 1}^{\outeriters}  \frac{1}{\mu} \left\| \subgrad (\outeriter-1) \right\|^2 \\
    &\leq -\sum_{\outeriter = 1}^{\outeriters }\frac{4\outeriter}{\mu(\outeriter+1)}\err[\outeriter - 1] + \sum_{\outeriter = 1}^{\outeriters}  \frac{1}{\mu} \left\| \subgrad (\outeriter-1) \right\|^2 \\
    & \leq -\sum_{\outeriter = 1}^{\outeriters }\frac{4\outeriter}{\mu(\outeriter+1)}\err[\outeriter - 1] + \frac{\outeriters \lipschitz[\val]}{\mu}\\
\end{align}
where the last line holds because the value function $\val$ is $\lipschitz[\val]$-Lipschitz continuous, where $\lipschitz[\val] = \max_{(\widehat{\outer}, \widehat{\inner}) \in \outerset \times \innerset} \left\| \grad[\outer] \obj\left( \widehat{\outer}, \inner^{*}(\widehat{\outer})\right) + \sum_{\numconstr = 1}^{\numconstrs} \langmult[\numconstr]^*(\widehat{\outer}, \inner^*(\widehat{\outer})) \grad[\outer] \constr[k]\left(\widehat{\outer}, \inner^{*}(\widehat{\outer})\right)\right\|$ and $\langmult^*(\widehat{\outer}, \inner^*(\widehat{\outer}))$ are the optimal KKT multipliers associated with $\inner^{*}(\widehat{\outer}) \in \argmax_{\inner \in \innerset:  \constr(\widehat{\outer}, \inner) \geq \zeros} \obj(\outer, \inner)$,  all the subgradients are bounded: i.e., for all $k \in \N, \left\| \subgrad (k-1) \right\| \leq \lipschitz[\val]$.

Additionally, note that by Cauchy-Schwarz we have, $|\err[\outeriters - 1]| \leq \left\|\subgrad (\outeriter) - \grad \val(\outer^{(\outeriter)})\right\| \left\|\outer^{(\outeriter)} - \outer^{*} \right\| \leq 2\lipschitz[\val] c$, where 
$c = \max_{\x, \x^\prime \in \outerset} \left\| \x - \x^\prime \right\|$.
\begin{align}
        \sum_{\outeriter = 1}^{\outeriters } \outeriter \left( \val (\outer^{(\outeriter - 1)}) - \val (\outer^{*}) \right) &\leq \sum_{\outeriter = 1}^{\outeriters } |\err[\outeriters - 1]|   \frac{4\outeriter}{\mu(\outeriter+1)} + \frac{\outeriters \lipschitz[\val]}{\mu}\\
        &\leq \sum_{\outeriter = 1}^{\outeriters } |\err[\outeriters - 1]|   \frac{4\outeriter}{\mu(\outeriter+1)} + \frac{\outeriters \lipschitz[\val]}{\mu}\enspace ,
\end{align}
Let $(\bestiter[x][t], \bestiter[y][t]) \in \argmin_{(\outer^{(k)}, \inner^{(k)}) : k \in [t]} \obj (\outer^{(k)}, \inner^{(k)})$.
Then:
\begin{align}
    \sum_{\outeriter = 1}^{\outeriters } \outeriter \left( \obj (\outer^{(\outeriter)}, \inner^{(\outeriter)}) - \obj (\outer^{*}, \inner^{(\outeriter)}) \right) &\leq \sum_{\outeriter = 1}^{\outeriters } |\err[\outeriters - 1]|   \frac{4\outeriter}{\mu(\outeriter+1)} + \frac{\outeriters}{\mu} \lipschitz[\val] \\
    \sum_{\outeriter = 1}^{\outeriters } \outeriter \left( \obj (\outer^{(\outeriter)}, \inner^{(\outeriter)}) - \max_{\inner \in \innerset: \constr (\outer^{*}, \inner) \geq \zeros} \obj (\outer^{*}, \inner) \right) &\leq \sum_{\outeriter = 1}^{\outeriters } |\err[\outeriters - 1]|   \frac{4\outeriter}{\mu(\outeriter+1)} +  \frac{\outeriters}{\mu} \lipschitz[\val] \\
    \left( \sum_{\outeriter = 1}^{\outeriters } \outeriter \right) \min_{\outeriter \in [\outeriters]} \left( \obj (\outer^{(\outeriter)}, \inner^{(\outeriter)}) - \max_{\inner \in \innerset: \constr (\outer^{*}, \inner) \geq \zeros} \obj (\outer^{*}, \inner) \right) &\leq \sum_{\outeriter = 1}^{\outeriters } |\err[\outeriters - 1]|   \frac{4\outeriter}{\mu(\outeriter+1)} + \frac{\outeriters}{\mu} \lipschitz[\val] \\
    \left( \sum_{\outeriter = 1}^{\outeriters } \outeriter \right)
    \left( \obj (\bestiter[x][\outeriters], \bestiter[y][\outeriters]) - \max_{\inner \in \innerset: \constr (\outer^{*}, \inner) \geq \zeros} \obj (\outer^{*}, \inner) \right) &\leq \sum_{\outeriter = 1}^{\outeriters } |\err[\outeriters - 1]|   \frac{4\outeriter}{\mu(\outeriter+1)} + \frac{\outeriters}{\mu} \lipschitz[\val] \\
    \left( \frac{(\outeriters + 1) \outeriters}{2} \right) \left( \obj (\bestiter[x][\outeriters], \bestiter[y][\outeriters]) - \max_{\inner \in \innerset: \constr (\outer^{*}, \inner) \geq \zeros} \obj (\outer^{*}, \inner) \right) &\leq \sum_{\outeriter = 1}^{\outeriters } |\err[\outeriters - 1]|   \frac{4\outeriter}{\mu(\outeriter+1)} + \frac{\outeriters \lipschitz[\val]}{\mu} \\
    \obj (\bestiter[x][\outeriters], \bestiter[y][\outeriters]) - \max_{\inner \in \innerset: \constr (\outer^{*}, \inner) \geq \zeros} \obj (\outer^{*}, \inner) &\leq O\left(\overline{\err[\outeriters - 1]}\right) + 
    \frac{\lipschitz[\val]}{\mu (\outeriters + 1)}  \enspace .
\end{align}
That is, as the number of iterations increases, the best iterate converges to a $(O\left(\overline{\err[\outeriters - 1]}\right), \delta)$-Stackelberg equilibrium.
Likewise, by the same logic we applied at the end of the proof as \Cref{l-mogd}, the best iterate converges to a $(\varepsilon + O\left(\overline{\err[\outeriters - 1]}\right), \delta)$-Stackelberg equilibrium in $O(\varepsilon^{-1})$ iterations.

\end{proof}

We now present a theorem which covers one of the cases given in \Cref{ngd-thms}. The proofs of the theorems that cover the other cases are similar to the proof below. We note that gradient ascent converges in $O(\varepsilon^{-1})$ iterations to a $\varepsilon$-maximum for a Lipschitz-smooth objective, and  in $O(\log(\varepsilon^{-1}))$ iterations to a $\varepsilon$-maximum for a Lipschitz-smooth and strongly-concave objective function \cite{boyd2004convex}.

\begin{theorem} \label{l-ngd}
Suppose \Cref{ngd} is run on a min-max Stackelberg game $(\outerset, \innerset, \obj, \constr)$ which satisfies \Cref{main-assum}. Suppose holds and that $\obj$ is $\lipschitz[\grad \obj]$-smooth. Let $(\bestiter[x][t], \bestiter[y][t]) \in \argmin_{(\outer^{(k)}, \inner^{(k)}) : k \in [t]} \obj (\outer^{(k)}, \inner^{(k)})$. For $\varepsilon \in (0,1)$, if we choose $\outeriters_\outer$ and $\outeriters_\inner$ s.t.\
%
    $\outeriters_\outer \geq N_{\outeriters_\outer}(\varepsilon) \in O(\varepsilon^{-2})$
    and 
    $\outeriters_\outer \geq N_{\outeriters_\inner}(\varepsilon) \in O(\varepsilon^{-1})$,
%
then there exists an iteration $\outeriters^{*} \leq \outeriters_\outer \outeriters_\inner = O(\varepsilon^{-3})$ s.t.\ $(\bestiter[x][\outeriters^{*}], \bestiter[y][\outeriters^{*}])$ is an $(\varepsilon, \varepsilon)$-Stackelberg equilibrium.
\end{theorem}

\begin{proof}[Proof of \Cref{l-mogd}]
Since $\obj$ is $\lipschitz[\grad \obj]$-smooth, it is well known that, for each outer iterate $\outer^{(\outeriter)}$, the inner gradient descent procedure returns an $\varepsilon$-maximum $\inner_{\epsilon}^{*}$ of $\obj (\outer^{(\outeriter)}, \inner)$ s.t.\ $\inner_{\epsilon}^{*} \in \innerset$ and $\constr(\outer^{(\outeriter)}, \inner_{\epsilon}^{*}) \geq \zeros$, in $O(\varepsilon^{-2})$ iterations \cite{boyd2004convex}. 
Combining the iteration complexity of the outer and inner loops using this result and \Cref{l-mogd}, we obtain an iteration complexity of $O(\varepsilon^{-2}) O(\varepsilon^{-1}) = O(\varepsilon^{-3})$.
\end{proof}
\section{An Economic Application: Details}\label{sec-app:fisher}

Our experimental goals were two-fold. 
First, we sought to understand the empirical convergence rate of our algorithms in different Fisher markets, in which the objective function in \Cref{fisher-program} satisfies different smoothness properties.
Second, we wanted to understand how the behavior of our two algorithms, max-oracle and nested gradient descent, differ in terms of the accuracy of the Stackelberg equilibria they find.

To answer these questions, we ran multiple experiments, each time recording the prices and allocations computed by
\Cref{mogd}, with an \emph{exact} max-oracle, and by \Cref{ngd}, with nested gradient ascent, during each iteration $t$ of the main (outer) loop.
For each run of each algorithm on each market with each set of initial conditions, we then computed the objective function's value for the iterates, i.e., $\obj(\outer^{(t)}, \inner^{(t)})$, which we plot in \Cref{fig:experiments1}.

\paragraph{Hyperparameters}
We randomly initialized 500 different linear, Cobb-Douglas, Leontief Fisher markets, each with $5$ buyers and $8$ goods.
Buyer $\buyer$'s budget $\budget[\buyer]$ was drawn randomly from a uniform distribution ranging from $100$ to $1000$ (i.e., $U[100,1000]$), while each buyer $\buyer$'s valuation for good $\good$, $\valuation[i][j]$, was drawn randomly from $U[5,15]$.
We ran both algorithms for 500, 300, and 700 iterations%
\footnote{In Algorithm~\ref{mogd:fm}, $\outeriters \in \{500, 300, 700\}$, while in Algorithm~\ref{ngd:fm}, $\outeriters_\price \in \{500, 300, 700\}$.}
for linear, Cobb-Douglas, and Leontief Fisher markets, respectively.
We started both algorithms from two sets of initial conditions, one with high prices (drawn randomly $U[50,55]$), and a second, with low prices (drawn randomly from $U[5,15]$).
%
We opted for a learning rate of 5 for both algorithms, after manual hyper-parameter tuning, and picked a decay rate of $t^{-\nicefrac{1}{2}}$, based on our theory, so that $\learnrate[1] = 5, \learnrate[2] = 3.54, \learnrate[3] = 2.89, \learnrate[4] = 2.5, \learnrate[5] = 2.24, \ldots$.

\paragraph{Programming Languages, Packages, and Licensing}
We ran our experiments in Python 3.7 \cite{van1995python}, using NumPy \cite{numpy}, Pandas \cite{pandas}, and CVXPY \cite{diamond2016cvxpy}.
\Cref{fig:experiments1} was graphed using Matplotlib \cite{matplotlib}.
To run the first order James test, we imported the data generated by our Python code into R \cite{R}.
Our R script manipulated the data using the Tidyverse package \cite{tidyverse}, and obtained the desired $p$-values using the STests package \cite{stestsR}.

R as a package is licensed under GPL-2 | GPL-3. Python software and documentation are licensed under the PSF License Agreement. Numpy is distributed under a liberal BSD license. Pandas is distributed under a new BSD license. Matplotlib only uses BSD compatible code, and its license is based on the PSF license. CVXPY is licensed under an APACHE license. Tidyverse is distributed under an MIT license.

\paragraph{Implementation Details}
In our execution of \Cref{mogd} for linear, Cobb-Douglas, and Leontief Fisher markets, we used an exact Max-Oracle, since there is a closed-form solution for the demand correspondence in these markets \cite{mas-colell}.

In our execution of \Cref{ngd}, in order to project each computed allocation onto the consumers' budget set, i.e., $\{\allocation \in \R^{\numbuyers \times \numgoods}_+ \mid \allocation\price \leq \budget\}$, we used the alternating projection algorithm \cite{boyd2004convex} for convex sets, and alternatively projected onto the sets $\R^{\numbuyers \times \numgoods}_+$ and $\{\allocation \in \R^{\numbuyers \times \numgoods} \mid \allocation\price \leq \budget\}$.

\paragraph{Computational Resources}
Our experiments were run on MacOS machine with 8GB RAM and an Apple M1 chip, and took about 2 hours to run. Only CPU resources were used.

\paragraph{Code Repository}
The data our experiments generated, as well as the code used to produce our visualizations and run the statistical tests, can be found in our \href{\rawcoderepo}{{\color{blue}code repository}}.

\subsection{Fisher Market Algorithms}

\begin{algorithm}
\caption{$\delta$-Approximate \emph{T\^atonnement} for Fisher Markets}
\label{mogd:fm}
\textbf{Inputs:} $\consumptions, \util, \budget, \learnrate, \outeriters, \price^{(0)}, \delta$ \\ 
\textbf{Output:} $(\allocation^*, \price^*)$
\begin{algorithmic}[1]
\For{$\outeriter = 1, \hdots, \outeriters$}
    \State For all $\buyer \in \buyers$, find $\allocation[\buyer]^{(t)}$ s.t.\ $\util[\buyer](\allocation[\buyer]^{(t)}) \geq \max_{\allocation[\buyer] : \allocation[\buyer] \cdot \price^{(\outeriter -1 )} \leq \budget[\buyer]} \util[\buyer](\allocation[\buyer] ) - \delta$ \& $\allocation[\buyer]^{(t)} \cdot \price^{(\outeriter -1 )} \leq \budget[\buyer]$
    \State Set $\price^{(\outeriter)} = \max \left\{   \price^{(\outeriter-1)} - \learnrate[t](1 - \sum_{\buyer \in \buyers} \allocation[\buyer]^{(t)}), 0 \right\}$
\EndFor
\State \Return $(\allocation^{(\outeriters)}, \price^{(\outeriters)})$
\end{algorithmic}
\end{algorithm}

\begin{algorithm}
\caption{$\delta$-Approximate Nested \emph{T\^atonnement} for Fisher Markets}
\label{ngd:fm}
\textbf{Inputs:} $\consumptions, \util, \budget, \learnrate[][\price], \learnrate[][\allocation] \outeriters_\price, \outeriters_\allocation, \price^{(0)}$ \\ 
\textbf{Output:} $(\allocation^*, \price^*)$
\begin{algorithmic}[1]
\For{$\outeriter = 1, \hdots, \outeriters_\price$}
    \For{$s = 1, \hdots, \outeriters_\allocation$}
        \State For all $\buyer \in \buyers$, $\allocation[\buyer]^{(t)} = \project[{\left\{\allocation[ ]: \allocation[ ] \cdot \price^{(\outeriter -1 )} \leq \budget[\buyer] \right\}}] 
        \left( \allocation[\buyer]^{(t)} + \learnrate[s][\allocation] \frac{\budget[\buyer]}{\util[\buyer](\allocation[\buyer]^{(t)})} \grad[{\allocation[\buyer]}] \util[\buyer](\allocation[\buyer]^{(t)}) \right)$
    \EndFor
    \State Set $\price^{(\outeriter)} = \max\left\{   \price^{(\outeriter-1)} - \learnrate[t][\price](1 - \sum_{\buyer \in \buyers} \allocation[\buyer]^{(t)}), 0 \right\}$
\EndFor
\State \Return $(\allocation^{(\outeriters)}, \price^{(\outeriters)})$ 
\end{algorithmic}
\end{algorithm}
\section{Additional Related Work}
\label{sec-app:related}

Much progress has been made recently in solving min-max games (with independent action sets), both in the convex-concave case and in non-convex-concave case.
For the former case, when $\obj$ is $\mu_\outer$-strongly-convex in $\outer$ and $\mu_\inner$-strongly-concave in $\inner$, \citeauthor{tseng1995variational} \cite{tseng1995variational}, \citeauthor{nesterov2006variational} \cite{nesterov2006variational}, and \citeauthor{gidel2020variational} \cite{gidel2020variational} proposed variational inequality methods, and \citeauthor{mokhtari2020convergence} \cite{mokhtari2020convergence}, gradient-descent-ascent (GDA)-based methods, all of which compute a solution in $\tilde{O}(\mu_\inner + \mu_\outer)$ iterations.
These upper bounds were recently complemented by the lower bound of $\tilde{\Omega}(\sqrt{\mu_\inner \mu_\outer})$, shown by \citeauthor{ibrahim2019lower}  \cite{ibrahim2019lower} and  \citeauthor{zhang2020lower}   \cite{zhang2020lower}.
Subsequently, \citeauthor{lin2020near}  \cite{lin2020near} and   \citeauthor{alkousa2020accelerated} \cite{alkousa2020accelerated} analyzed algorithms that converge in $\tilde{O}(\sqrt{\mu_\inner \mu_\outer})$ and $\tilde{O}(\min\left\{\mu_\outer \sqrt{\mu_\inner}, \mu_\inner \sqrt{\mu_\outer} \right\})$ iterations, respectively. 

For the special case where $\obj$ is $\mu_\outer$-strongly convex in $\outer$ and linear in $\inner$, \citeauthor{juditsky2011first}  \cite{juditsky2011first},  \citeauthor{hamedani2018primal}  \cite{hamedani2018primal}, and \citeauthor{zhao2019optimal}  \cite{zhao2019optimal} all present methods that converge to an $\varepsilon$-approximate solution in $O(\sqrt{\nicefrac{\mu_\outer}{\varepsilon}})$ iterations.
When the strong concavity or linearity assumptions of $\obj$ on $\inner$ are dropped, and
$\obj$ is assumed to be $\mu_\outer$-strongly-convex in $\outer$ but only concave in $\inner$, \citeauthor{thekumparampil2019efficient} \cite{thekumparampil2019efficient} provide an algorithm that converges to an $\varepsilon$-approximate solution in $\tilde{O}(\nicefrac{\mu_\outer}{\varepsilon})$ iterations, and \citeauthor{ouyang2018lower} \cite{ouyang2018lower} provide a lower bound of $\tilde{\Omega}\left(\sqrt{\nicefrac{\mu_\outer}{\varepsilon}}\right)$ iterations on this same computation.
\citeauthor{lin2020near} then went on to develop a faster algorithm, with iteration complexity of $\tilde{O}\left(\sqrt{\nicefrac{\mu_\outer}{\varepsilon}}\right)$, under the same conditions.

When $\obj$ is simply assumed to be convex-concave, \citeauthor{nemirovski2004prox} \cite{nemirovski2004prox}, \citeauthor{nesterov2007dual} \cite{nesterov2007dual}, and \citeauthor{tseng2008accelerated} \cite{tseng2008accelerated} describe algorithms that solve for an $\varepsilon$-approximate solution with $\tilde{O}\left(\varepsilon^{-1}\right)$ iteration complexity, and \citeauthor{ouyang2018lower} \cite{ouyang2018lower} prove a corresponding lower bound of $\Omega(\varepsilon^{-1})$.

When $\obj$ is assumed to be non-convex-$\mu_\inner$-strongly-concave, and the goal is to compute a first-order Nash, \citeauthor{sanjabi2018stoch} \cite{sanjabi2018stoch} provide an algorithm that converges to $\varepsilon$-an approximate solution in $O(\varepsilon^{-2})$ iterations.
\citeauthor{jin2020local} \cite{jin2020local}, \citeauthor{rafique2019nonconvex} \cite{rafique2019nonconvex}, \citeauthor{lin2020gradient} \cite{lin2020gradient}, and \citeauthor{lu2019block} \cite{lu2019block} provide algorithms that converge in $\tilde{O}\left(\mu_\inner^2 \varepsilon^{-2}\right)$ iterations, while \citeauthor{lin2020near} \cite{lin2020near} provide an even faster algorithm, with an iteration complexity of $\tilde{O}\left(\sqrt{\mu_\inner} \varepsilon^{-2}\right)$.


When $\obj$ is non-convex-non-concave and the goal to compute is an approximate first-order Nash equilibrium, \citeauthor{lu2019block} \cite{lu2019block} provide an algorithm with iteration complexity $\tilde{O}(\varepsilon^{-4})$, while \citeauthor{nouiehed2019solving} \cite{nouiehed2019solving} provide an algorithm with iteration complexity $\tilde{O}(\varepsilon^{-3.5})$. More recently, \citeauthor{ostrovskii2020efficient} \cite{ostrovskii2020efficient} and \citeauthor{lin2020near} \cite{lin2020near} proposed an algorithm with iteration complexity $\tilde{O}\left(\varepsilon^{-2.5}\right)$.

When $\obj$ is non-convex-non-concave and the desired solution concept is a ``local'' Stackelberg equilibrium, \citeauthor{jin2020local} \cite{jin2020local}, \citeauthor{rafique2019nonconvex} \cite{rafique2019nonconvex}, and \citeauthor{lin2020gradient} \cite{lin2020gradient} provide algorithms with a $\tilde{O}\left( \varepsilon^{-6} \right)$ complexity.
More recently, \citeauthor{thekumparampil2019efficient} \cite{thekumparampil2019efficient}, \citeauthor{zhao2020primal} \cite{zhao2020primal}, and \citeauthor{lin2020near} \cite{lin2020near} have proposed algorithms that converge to an $\varepsilon$-approximate solution in $\tilde{O}\left( \varepsilon^{-3}\right)$ iterations.

We summarize the literature pertaining to the convex-concave and the non-convex-concave settings in Tables~\ref{tab:fixed-convex-concave} and~\ref{tab:fixed-nonconvex-concave}, respectively.

\renewcommand*\arraystretch{1.5}
\begin{table}[H]
    \centering
    \caption{Iteration complexities for min-max games (with independent strategy sets) in convex-concave settings. Note that these results assume that the objective function is Lipschitz-smooth.} \label{tab:fixed-convex-concave}
    \begin{tabular}{|c|c|c|}\hline
    Setting & Reference & Iteration Complexity \\ \hline
    \multirow{8}{*}{$\mu_\outer$-Strongly-Convex-$\mu_\inner$-Strongly-Concave} & \cite{tseng1995variational} & \multirow{4}{*}{$\tilde{O}\left( \mu_\outer + \mu_\inner\right)$} \\\cline{2-2}
         & \cite{nesterov2006variational}  & \\ \cline{2-2}
         & \cite{gidel2020variational}     & \\ \cline{2-2}
         & \cite{mokhtari2020convergence}  &  \\ \cline{2-3}
         & \cite{alkousa2020accelerated}   & $\tilde{O}\left(\min \left\{\mu_\outer \sqrt{\mu_\inner}, \mu_\inner \sqrt{\mu_\outer} \right\} \right)$\\ \cline{2-3}
         & \cite{lin2020near}              & $\tilde{O}(\sqrt{\mu_\outer \mu_\inner})$ \\ \cline{2-3}
         & \cite{ibrahim2019lower} & $\tilde{\Omega}(\sqrt{\mu_\outer \mu_\inner})$\\ \cline{2-2}
         & \cite{zhang2020lower} & \\ \hline \hline
    \multirow{3}{*}{$\mu_\outer$-Strongly-Convex-Linear}    & \cite{juditsky2011first} & \multirow{3}{*}{$O\left( \sqrt{\nicefrac{\mu_\outer}{\varepsilon}}\right)$} \\\cline{2-2}
    & \cite{hamedani2018primal} & \\\cline{2-2}
    & \cite{zhao2019optimal}& \\\hline \hline
    \multirow{3}{*}{$\mu_\outer$-Strongly-Convex-Concave} & \cite{thekumparampil2019efficient} & $\tilde{O}\left( \nicefrac{\mu_\outer }{\sqrt{\varepsilon}} \right)$ \\ \cline{2-3}
    & \cite{lin2020near} & $\tilde{O}(\sqrt{\nicefrac{\mu_\outer}{\varepsilon}})$ \\ \cline{2-3}
    & \cite{ouyang2018lower} & $\tilde{\Omega}\left( \sqrt{\nicefrac{\mu_\outer}{\varepsilon}}\right)$ \\ \hline \hline
    \multirow{5}{*}{Convex-Concave} & \cite{nemirovski2004prox} & \multirow{2}{*}{$O\left( \varepsilon^{-1}\right)$} \\ \cline{2-2}
    & \cite{nesterov2007dual} & \\ \cline{2-2}
    & \cite{tseng2008accelerated} & \\ \cline{2-3}
    & \cite{lin2020near} &  $\tilde{O}\left(\varepsilon^{-1}\right)$\\ \cline{2-3}
    & \cite{ouyang2018lower} & $\Omega(\varepsilon^{-1})$ \\ \hline 
    \end{tabular}
    \renewcommand*\arraystretch{1}
\end{table}

\begin{table}[H]
    \centering
    \caption{Iteration complexities for min-max games (with independent strategy sets) in non-convex-concave settings. Note that although all these results assume that the objective function is Lipschitz-smooth, some authors make additional assumptions: e.g., \cite{nouiehed2019solving} obtain their result for objective functions that satisfy the Lojasiwicz condition.}
    \label{tab:fixed-nonconvex-concave}
    \renewcommand*\arraystretch{1.5}
    \begin{tabular}{|c|c|c|}\hline
    Setting & Reference & Iteration Complexity\\ \hline
        \multirow{5}{*}{\makecell{Nonconvex-$\mu_\inner$-Strongly-Concave,\\ First Order Nash Equilibrium \\ or Local Stackelberg Equilibrium}} & \cite{jin2020local} & \multirow{4}{*}{$ \tilde{O}(\mu_\inner^2 \varepsilon^{-2})$} \\
         & \cite{rafique2019nonconvex} & \\ \cline{2-2}
         & \cite{lin2020gradient}  & \\ \cline{2-2}
         & \cite{lu2019block} & \\ \cline{2-3}
         & \cite{lin2020near} & $\tilde{O}\left( \sqrt{\mu_\inner} \varepsilon^{-2} \right)$\\ \hline \hline
        \multirow{4}{*}{\makecell{Nonconvex-Concave,\\ First Order Nash Equilibrium}} & \cite{lu2019block}  & $\tilde{O}\left(\varepsilon^{-4}\right)$ \\ \cline{2-3}
        & \cite{nouiehed2019solving} & $\tilde{O}\left( \varepsilon^{-3.5}\right)$ \\ \cline{2-3}
        & \cite{ostrovskii2020efficient} & \multirow{2}{*}{$\tilde{O}\left( \varepsilon^{-2.5}\right)$} \\ \cline{2-2}
        & \cite{lin2020near} &  \\ \hline \hline
        \multirow{6}{*}{\makecell{Nonconvex-Concave,\\ Local Stackelberg Equilibrium}} & \cite{jin2020local} & \multirow{3}{*}{$\tilde{O}(\varepsilon^{-6})$}\\  \cline{2-2}
        & \cite{nouiehed2019solving} & \\ \cline{2-2}
        & \cite{lin2020near} & \\ \cline{2-3}
        & \cite{thekumparampil2019efficient} & \multirow{3}{*}{$\tilde{O}(\varepsilon^{-3})$}\\ \cline{2-2}
        & \cite{zhao2020primal} & \\
        & \cite{lin2020near} & \\ \hline 
    \end{tabular}
    \renewcommand*\arraystretch{1}
\end{table}

\section{Future Directions}\label{sec_app:future}
Our experiments suggest that the smoothness properties of the value function determine the convergence speed to a Stackelberg equilibrium.
Additionally, our experiments with Leontief Fisher markets suggest that our convergence results could perhaps be generalized to convex-concave objective functions which are \emph{not\/} necessarily continuously differentiable; such a result would require a generalization of the subdifferential envelope theorem we have introduced.

Another worthwhile direction would be to try to derive stronger conditions under which the value function $\val$ is Lipschitz-smooth; under stronger assumptions we might be able to achieve faster convergence results.
In order for $\val$ to be differentiable, the subdifferential given by \Cref{envelope-sd} would have to be a singleton, for all $\outer \in \outerset$.
This would require not only that the solution function $\innerset^{*}(\outer)$ be a singleton for all $\outer \in \outers$, but that the optimal KKT multipliers $ \langmults(\widehat{\outer}, \inner^{*}(\widehat{\outer})))$ were also unique.
The former can be guaranteed when $\obj$ is strictly concave in $\inner$, while the latter condition is satisfied when the Linear Independence Constraint Qualification 
condition \cite{henrion1992constraint} holds.
However, even when both these conditions hold, and when both the objective function and the (parameterized) constraints (i.e., $\constr (\outer, \inner)$) are Lipschitz-smooth, the value function is not guaranteed to be Lipschitz-smooth. 

\citeauthor{gao2020polygm} \cite{gao2020polygm} have proposed first-order methods for the efficient computation of competitive equilibria in large Fisher markets.
It would be worth exploring whether our view of competitive equilibria in Fisher markets as min-max Stackelberg games can facilitate the extension of their fast convergence results to a larger class of Fisher markets.

A question of interest both for Fisher market dynamics and convex-concave min-max Stackelberg games is whether gradient-descent-ascent (GDA) converges in the dependent strategy set setting as it does in the independent strategy setting \cite{lin2020gradient}.
GDA dynamics for Fisher markets correspond to myopic best-response dynamics,
which are of general economic interest (see, for example, \cite{monderer1996potential}).

Finally, our results at present concern only \emph{convex-concave\/} min-max Stackelberg games.
It would be of interest to extend these results to the non-convex-concave setting.
Doing so could improve our understanding of competitive equilibria in Fisher markets with non-homogeneous utility functions, and our understanding of optimal auctions.
Recently, \citeauthor{dai2020optimality} \cite{dai2020optimality} defined the concept of a local minimax point, whose analog in Stackelberg games would be a local Stackelberg equilibria.
We believe that any algorithm designed for a non-convex-concave setting would have to aim for a local solution concept, as it is unlikely that (global) Stackelberg equilibria are computable in polynomial time in non-convex-concave settings, as non-convex optimization is hard.

\section*{Broader Impact}

Our work can be used to 
expand the scope of many machine learning techniques as discussed in our introduction.
For example, min-max optimization has been key to the development of fairer classifiers in recent years.
On the other hand, our methods could also be used to train GANs, which have been used to create deepfakes with malicious goals. 
Our work can also be used to improve economic outcomes for companies running online marketplaces, as Fisher markets have been applied to resource allocation problems---specifically, fair division problems---in recent years.
Moreover, with sufficient oversight, progress can be easily monitored, by following the trend of the the objective function across iterations, so that inappropriate use can be detected. 
\end{document}